\documentclass[10pt, journal,twocolumn]{IEEEtran}
\pdfoutput=1
\usepackage[font=small]{caption}
\usepackage{amsmath}
\usepackage{mathtools}
\usepackage{amsmath}
\usepackage{amssymb}
\usepackage{cite}
\usepackage{array}
\usepackage{amsthm}
\usepackage{amsmath}
\usepackage{tabulary}
\usepackage{multirow}

\usepackage{subfig}
\IEEEoverridecommandlockouts
\usepackage{pifont}
\usepackage{epsfig}
\usepackage{graphicx}
\usepackage{url}
\usepackage{amssymb}
\usepackage{pgfplots}
\usepackage{bbm}
\usepackage{paralist} %

\newcounter{MYtempeqncnt}
\newcommand{\expect}[1]{{\mathbb{E}\left[{#1}\right]}}

\newcommand{\cexpect}[2]{{\mathbb{E}_{#2}\left[{#1}\right]}}
\newcommand{\pr}{{\mathbb{P}}}
\newcommand{\ppr}{{\mathbb{P}^o}}

\newcommand{\mgf}{\mathcal{L}}

\newcommand{\loadpmf}{\mathsf{K}}

\newcommand{\varpl}{\xi^2}
\newcommand{\sdpl}{\xi}
\newcommand{\meen}{m}

\newcommand{\pcf}{\epsilon}

\newcommand{\ple}{\alpha}
\newcommand{\rple}{\alpha}

\newcommand{\power}{\mathrm{P}}
\newcommand{\res}{\mathrm{B}}

\newcommand{\fsldb}{\beta}

\newcommand{\noisepower}{\sigma^2_{N}}

\newcommand{\dnsty}{{\lambda}}
\newcommand{\wfrac}{\omega}
\newcommand{\loadfrac}{\kappa}

\newcommand{\dnstyr}{\Lambda}
\newcommand{\ndnstyr}{\mathrm{M}}

\newcommand{\userdnsty}{\lambda_u}

\newcommand{\SINRthresh}{\tau} 
\newcommand{\tmin}{\tau_\mathrm{min}}

\newcommand{\Q}{\mathrm{Q}}

\newcommand{\RATEthresh}{\rho}

\newcommand{\nRATEthresh}{\hat{\rho}}
 
\newcommand{\uRATEthresh}{v}

\newcommand{\real}[1]{\mathbb{R}^{#1}}

\newcommand{\uth}[1]{#1^\text{th}}

\newcommand{\utmoment}{\bar{\zeta} }
\newcommand{\ltmoment}{\underline{\zeta} }

\newcommand{\SINR}{\mathtt{SINR}}
\newcommand{\pl}{L}
\newcommand{\plrnd}{\chi}
\newcommand{\SIR}{\mathtt{SIR}}
\newcommand{\SNR}{\mathtt{SNR}}
\newcommand{\INR}{\mathtt{INR}}
\newcommand{\totI}{I_t}

\newcommand{\intrfrnc}{I}
\newcommand{\IF}{J}

\newcommand{\assocr}{\mathcal{C}}
\newcommand{\rate}{\mathtt{Rate}}

\newcommand{\PPP}{\Phi}
\newcommand{\BSP}{\Phi}
\newcommand{\WBSP}{\Phi_w}

\newcommand{\PropPP}{\mathcal{N}}

\newcommand{\load}[1]{{N}_{#1}}

\newcommand{\avload}[1]{\bar{N}_{#1}} 
 
\newcommand{\chanl}{H}

\newcommand{\sha}{S}

\newcommand{\indic}{\mathbbm{1}}
\newcommand{\sbs}{X^*}
\newcommand{\swbs}{Y^*}

\newcommand{\pcov}{\mathcal{P}}
\newcommand{\scov}{\mathcal{S}}

\newcommand{\rcov}{\mathcal{R}}

\newcommand{\passoc}{\mathcal{A}}

\newcommand{\plos}{\mathrm{C}} 
\newcommand{\af}{\eta}

\newcommand{\dlos}{\mathrm{D}}

\newcommand{\again}{\psi}

\newcommand{\antgain}{\mathrm{G}} 
\newcommand{\bwidth}{\theta_b} 
\newcommand{\antgainmax}{\mathrm{G_{max}}}

\newtheorem{thm}{{\bf Theorem}}
\newtheorem{cor}{Corollary}
\newtheorem{rem}{Remark}
\newtheorem{lem}{Lemma}
\newtheorem{prop}{Proposition}

\theoremstyle{definition}

\theoremstyle{thm}
\newtheorem{asmptn}{Assumption}

\begin{document}

\title{Tractable Model for Rate  in Self-Backhauled Millimeter Wave Cellular Networks}

\author{\IEEEauthorblockN{Sarabjot Singh, Mandar N. Kulkarni, Amitava Ghosh, and  Jeffrey G. Andrews }

\thanks{This work has been supported by   Nokia.   A part of this paper is accepted for presentation at IEEE Asilomar, November 2014 in Pacific Grove, CA, USA \cite{SinKulAndConf14}.

 M. N. Kulkarni ({\tt mandar.kulkarni@utexas.edu})  and J. G. Andrews ({\tt jandrews@ece.utexas.edu}) are with Wireless Networking and Communications Group (WNCG), The University of Texas at Austin, USA. S. Singh ({\tt sarabjot.singh@nokia.com}) and A. Ghosh ({\tt amitava.ghosh@nsn.com}) are with Nokia at Berkeley and at Arlington Heights respectively. S. Singh was also with WNCG. 
}}
\maketitle

\begin{abstract}
Millimeter wave (mmWave) cellular systems will require high gain directional antennas  and dense base station (BS) deployments to overcome high near field path loss and poor diffraction. As a desirable side effect, high gain antennas offer interference isolation, providing an opportunity to incorporate \textit{self-backhauling}--BSs backhauling among themselves in a mesh architecture without significant loss in throughput--to enable the requisite large BS densities.  The use of directional antennas  and resource sharing between access and backhaul links leads to  coverage and rate trends  that differ significantly from  conventional ultra high frequency (UHF)  cellular systems.  In this paper, we propose a general and tractable mmWave cellular model capturing these key trends and characterize the associated rate distribution. The developed model and analysis is validated using actual building locations from  dense urban settings and empirically-derived path loss models. The analysis shows that   in sharp contrast to the interference-limited nature of UHF cellular networks, the \emph{spectral efficiency}   of mmWave networks  (besides total rate) also increases with BS density particularly at the cell edge.  Increasing the system bandwidth, although boosting median and peak rates, does not significantly influence the cell edge rate. With self-backhauling, different combinations of the wired backhaul fraction (i.e. the fraction of BSs with a wired connection) and  BS density are shown to guarantee the same median rate (QoS).
\end{abstract}

\section{Introduction}
The scarcity of ``beachfront" UHF ($300$ MHz- $3$ GHz) spectrum   and surging wireless traffic demands has made going higher in  frequency for terrestrial communications inevitable.  The capacity boost provided by   increased Long Term Evolution (LTE) deployments  and  aggressive small cell, particularly Wi-Fi, offloading  has, so far, been able to cater to the increasing traffic demands,  but to meet  the projected \cite{cisco} traffic needs of $2020$ (and beyond)  availability of large amounts of  new spectrum would be indispensable. The only place where a significant amount of unused or lightly used  spectrum is available is in the millimeter wave (mmWave) bands ($20-100$ GHz).  With many GHz of spectrum to offer, mmWave bands are becoming increasingly attractive as one of the front runners  for the next generation (a.k.a. ``$5$G")  wireless cellular networks \cite{BocHea14,Rap13,And5G14}. 

\subsection{Background and recent work}
\textbf{Feasibility of mmWave cellular.} Although mmWave based indoor and personal area networks  have already received considerable  traction \cite{baykas11,Dan60},  such frequencies have long  been  deemed unattractive  for cellular communications primarily due to the large near-field loss and poor penetration  (blocking) through  concrete, water, foliage, and other common material.   Recent  research efforts \cite{PiKha, Rap13,RohMag14,RapTransAnt13,RanRapErk14,AkdRanRap13,Larew13,GhoshJSAC} have, however,  seriously challenged  this widespread perception. In principle,  the  smaller wavelengths associated with mmWave allow  placing many  more miniaturized antennas in the same physical area, thus compensating for the  near-field path loss \cite{PiKha,RohMag14}. Communication ranges of 150-200m  have   been shown to be  feasible  in dense urban scenarios  with the use of such high gain directional antennas \cite{Rap13,RapTransAnt13,RohMag14}.  Although mmWave signals do indeed penetrate and diffract  poorly through urban clutter, dense urban environments offer rich multipath (at least for outdoor) with strong reflections; making non-line-of-sight (NLOS) communication   feasible with  familiar path loss exponents in the range of  $3$-$4$ \cite{Rap13,RohMag14}. 
Dense and directional mmWave networks  have been shown to exhibit a  similar spectral efficiency to   $4$G (LTE) networks (of the same density) \cite{RanRapErk14,AkdRanRap13}, and hence can achieve an order of magnitude gain in  throughput  due to the increased  bandwidth.

\textbf{Coverage trends in mmWave cellular.} With  high gain  directional antennas and newfound sensitivity to blocking, mmWave coverage trends will be quite different from previous cellular networks. Investigations via detailed system level simulations   \cite{RanRapErk14,AkdRanRap13, AboChar13,Larew13,GhoshJSAC,KulGC14} have shown large bandwidth mmWave  networks  in   urban settings\footnote{Note that  capacity crunch is also most severe in such dense urban scenarios.}  tend to be noise limited--i.e.  thermal noise dominates interference--in contrast to $4$G cellular networks, which are usually strongly interference limited.  As a result, mmWave outages are mostly due to a  low signal-to-noise-ratio ($\SNR$) instead of low signal-to-interference-ratio ($\SIR$).  This insight was also highlighted in  an earlier work  \cite{Madhow11} for   directional   mmWave ad hoc networks.  Because cell edge users    experience low  $\SNR$ and are power limited, increased  bandwidth leads to little or no gain in their rates as compared to the median or peak rates \cite{AkdRanRap13}. Note that rates were compared with a $4$G network in \cite{AkdRanRap13}, however, in this paper we also investigate the effect of  bandwidth on rate in mmWave regime. 

\textbf{Density and backhaul.} As highlighted in   \cite{Larew13,RanRapErk14,AkdRanRap13,PiKha,AboChar13,GhoshJSAC}, dense BS deployments are essential for  mmWave networks to achieve acceptable coverage and rate. This  poses a particular challenge for the backhaul network, especially given the huge rates stemming from mmWave bandwidths on the order of GHz.   However,   the interference isolation provided by narrow directional beams    provides a  unique opportunity    for    scalable  backhaul architectures \cite{PiKha,IDBackhaul,BackhaulSS14}.  Specifically, \textit{self-backhauling} is a natural and scalable solution \cite{IDBackhaul,BackhaulSS14,kim2014system},   where BSs with   wired backhaul provide for the backhaul of BSs without it using a mmWave link. This architecture is quite different from the mmWave based point-to-point backhaul \cite{Eric_CoMag} or the relaying architecture \cite{steven2009relay} already in use, as \begin{inparaenum}[(a)]  \item the BS with wired backhaul serves multiple BSs,  and \item access and backhaul links share the total pool of available resources at each BS\end{inparaenum}. This results in a multihop network, but one in which the hops need not interfere, which is what largely doomed previous attempts at mesh networking.  However,  both the load on the backhaul  and access link impact the eventual user rate, and
 a general and tractable model that integrates the backhauling architecture into the analysis of a mmWave cellular network seems important to develop.  The main objective  of this work is to address this. As we show,   the very notion of a coverage/association cell  is strongly questionable  due to the sensitivity of mmWave to blocking in dense urban scenarios. Characterizing the load and   rate  in such networks, therefore, is non-trivial   due to the formation of irregular and ``chaotic" association cells (see Fig. \ref{fig:ManAssoc}).

\textbf{Relevant models.} Recent work in developing models for the analysis of mmWave  cellular networks (ignoring  backhaul) includes   \cite{AkoHea,BaiHea14,BaiHeaSIP}, where the downlink $\SINR$ distribution is characterized assuming BSs to be spatially distributed according to a Poisson point process (PPP).   No blockages were assumed  in \cite{AkoHea}, while  \cite{BaiHea14} proposed a framework to derive $\SINR$ distribution with an isotropic blockage model, and derived the expressions for  a line of sight (LOS) ball based blockage model  in which all nearby BSs were assumed LOS and all BSs beyond a certain distance from the user were ignored.   This LOS ball blockage model can be interpreted as a step function approximation of the  exponential blockage model proposed in  \cite{BaiVazHea} and used in \cite{BaiHeaSIP}. The randomness in the distance-based path loss (shown  to be quite significant in empirical studies \cite{GhoshJSAC}), was however ignored in prior analytical works.  Coverage was shown \cite{BaiHea14}  to  exhibit a non-monotonic trend with BS density. In this work, however, we show that if the finite user population is taken into account (ignored in \cite{BaiHea14}), $\SINR$ coverage  may increase monotonically  with density.     Although characterizing $\SINR$ is important, rate is the key metric, and can follow quite different trends \cite{SinDhiAnd13,AndLoadCommag13} than $\SINR$ because the user load is essentially a pre-log factor  whereas $\SINR$ is inside the log in the Shannon capacity formula. 
\subsection{Contributions}
The major contributions of this paper can be categorized broadly as follows:\\
\textbf{Tractable mmWave  cellular model}. A tractable and general model is proposed in Sec. \ref{sec:sysmodel} for characterizing   coverage and   rate distribution in self-backhauled mmWave  cellular networks. The proposed blockage model allows for an adaptive fraction of area around each user to be LOS.  Assuming the BSs are distributed according to   a PPP, the analysis, developed in Sec. \ref{sec:analysis},   accounts for   different path losses (both mean and variance) of LOS/NLOS links  for both access and backhaul--consistent with empirical studies \cite{Rap13,GhoshJSAC}. We identify and characterize two types of association cells in self-backhauled networks: \begin{inparaenum}[(a)] \item  \textit{user association area} of a BS which impacts the load on the access link, and \item \textit{BS association area} of a BS with wired backhaul required for quantifying  the load on the backhaul link. \end{inparaenum}  The rate distribution  across the entire network, accounting for the random backhaul and  access link capacity,   is then characterized  in Sec \ref{sec:analysis}. Further, the analysis is extended  to derive  the rate distribution with offloading to and from a co-existing UHF macrocellular network.\\
\textbf{Validation of model and analysis}. In Sec. \ref{sec:validation}, the analytical rate distribution  derived from  the proposed model is compared with that obtained from simulations employing actual building locations in dense urban regions of    New York  and Chicago \cite{KulGC14}, and empirically measured path loss models \cite{GhoshJSAC}. The demonstrated  close match between  the analysis and simulation  validates the proposed blockage model and our analytical approximation of  the irregular association areas and load.\\
\textbf{Performance insights.} Using the developed framework, it is demonstrated in Sec. \ref{sec:results} that: 
\begin{itemize}
\item MmWave  networks in dense urban scenarios employing high gain narrow beam antennas tend to be noise-limited for  ``moderate"   BS densities. Consequently,  densification of the network improves the  $\SINR$ coverage, especially for uplink.  
Incorporating the impact of finite user density, $\SINR$ coverage  can possibly  increase with density even in the very large density regime.
\item Cell edge users experience poor $\SNR$ and hence are particularly power limited. Increasing the air interface bandwidth, as a result, does not significantly improve the cell edge rate, in contrast to the  cell median or peak rates. Improving the density, however, improves the cell edge rate drastically.  Assuming all users to be mmWave capable, cell edge rates are also shown to improve by  reverting users to the UHF network whenever reliable mmWave communication  is unfeasible.
\item Self-backhauling is  attractive  due to  the diminished  effect of interference in  such networks. Increasing the fraction of BSs with wired backhaul, obviously, improves the peak rates in the network. Increasing the density of BSs while keeping the density of wired backhaul BSs constant in the network, however, leads to  saturation of user rate coverage. We characterize the corresponding \textit{saturation density} as the BS density beyond which marginal improvement in rate coverage would be observed without further wired backhaul provisioning.  The saturation density  is shown to be proportional to the density of BSs with wired backhaul.
\item  The same  rate coverage/median rate is shown to be achievable with  various combinations of (i) the fraction of wired backhaul BSs and (ii) the density of BSs.  A rate-density-backhaul contour is characterized, which  shows, for example, that  the same median rate can be achieved  through a higher fraction of wired backhaul BSs in sparse networks or a  lower fraction of  wired backhaul BSs in dense deployments.
\end{itemize}

\section{System Model}\label{sec:sysmodel}
\subsection{Spatial locations}
The  BSs  in the network are assumed to be distributed uniformly in $\real{2}$ as a  homogeneous  PPP  of density (intensity) $\dnsty_t$.  
 The PPP assumption is adopted  for  tractability, however other spatial models can be expected to exhibit similar trends due to the nearly constant  $\SINR$ gap over that of the  PPP \cite{ADG_Tcom}. 
A fraction $\frac{\mu}{\dnsty_t}$ and $\frac{\dnsty}{\dnsty_t}$ (assigned by independent marking, with $\mu+ \dnsty=\dnsty_t$) of the BSs are assumed to form the UHF macrocellular and mmWave  network respectively, and thus the  corresponding (independent) PPPs are:  $\PPP_\mu$ with density $\mu$ and $\PPP$  with density $\lambda$ respectively.
The users are also assumed to be uniformly distributed  as a  PPP $\PPP_u$ of density (intensity) $\userdnsty$ in $\real{2}$. 
 A  fraction  $\wfrac$ of the mmWave BSs (called anchored BS or A-BS  henceforth) have  wired backhaul and the rest of  mmWave BSs backhaul wirelessly to   A-BSs.  So,  the A-BSs      \textit{serve} the rest of the BSs in the network resulting in   two-hop links to the users associated with the BSs.  Independent \textit{marking}  assigns  wired backhaul (or not) to each mmWave BS and hence the resulting independent  point process of A-BSs $\PPP_w$ is   also a PPP   with density $\dnsty\wfrac$. 

 Notation   is summarized in Table \ref{tbl:param}. Capital roman font is used for parameters and italics for random variables.
 
\subsection{Propagation assumptions}
For   mmWave transmission, the power received at $y \in \real{2}$ from  a transmitter at $x \in\real{2} $ transmitting with power $\power(x)$ is given by $\power(x)\again(x,y)\pl(x,y)^{-1}$, where  $\again$ is the combined antenna gain  of the receiver and transmitter and  $\pl$ (dB)$  = \fsldb + 10\rple\log_{10}\|x-y\|+ \plrnd$ is the associated path loss in dB, where $\plrnd \sim \mathcal{N}(0,\varpl)$. Different strategies  can be adopted  for formulating  the path loss model from field measurements. If $\fsldb$ is constrained to be the path loss at a close-in reference distance, then $\rple$ is  physically interpreted as  the path loss exponent. But if these parameters  are obtained by a best linear fit,  then $\fsldb$ is the intercept and $\rple$ is the slope of the fit, and  no physical interpretation may be ascribed. The  deviation in fitting (in dB scale) is modeled as a zero mean Gaussian (Lognormal in linear scale) random variable  $\plrnd$ with variance   $\varpl$.     
Motivated by the studies in \cite{GhoshJSAC,Rap13}, which point to different LOS and NLOS path loss parameters for access (BS-user) and backhaul (BS-A-BS) links, the analytical model in this paper accommodates distinct $\fsldb$,  $\ple$, and $\varpl$ for each.    Each mmWave BS and user is assumed to  transmit  with power    $\power_b$ and $\power_u$, respectively,  over a bandwidth $\res$. The transmit power and bandwidth for  UHF BS  is denoted  by $\power_\mu$ and  $\res_\mu$ respectively.

All mmWave BSs are assumed to be equipped with directional antennas with a sectorized gain pattern. Antenna gain pattern for a BS as a function of angle $\theta$ about the steering  angle is given by 
\begin{equation*}
\antgain_b(\theta)= 
\begin{cases} \antgain_{\mathrm{max}} & \text{if $ |\theta| \leq \bwidth$}
\\\antgain_{\mathrm{min}} &\text{otherwise.}
\end{cases},
\end{equation*}
where $\bwidth$ is the beam-width or main lobe width. Similar abstractions have been used in  the prior study of directional ad hoc networks \cite{YiPei03}  and recently   mmWave networks \cite{AkoHea,BaiHea14}. The user antenna gain pattern $G_u(\theta)$ can be modeled in the same manner; however, in this paper we assume omnidirectional antennas for the users. The  beams of all non-intended links are assumed to be  randomly oriented with respect to each other  and hence the effective antenna gains (denoted by $\again$) on the interfering links are random.   The antennas beams  of the intended access  and backhaul link  are assumed to be aligned, i.e., the effective gain on the desired access link is $\antgain_{\mathrm{max}}$ and on the desired backhaul link is $\antgain_{\mathrm{max}}^2$. Analyzing the impact of alignment errors on the desired link is beyond the scope of the current work, but can be done on the lines of the recent work  \cite{WildmanNLW13}. It is worth pointing out here that   since our analysis is restricted to $2$-D, the directivity of the antennas is modeled only in the azimuthal plane, whereas in practice due to the $3$-D antenna gain pattern \cite{GhoshJSAC,RohMag14}, the RF isolation to the unintended receivers would also be provided by differences in elevation angles. 

\begin{table}
\caption{Notation and simulation parameters}
	\label{tbl:param}
  \begin{tabulary}{\columnwidth}{ |C | L | L|}
    \hline
  \textbf{ Notation} & \textbf{Parameter} & \textbf{Value (if applicable) }\\\hline
 $\PPP$, $\dnsty$ & mmWave BS PPP and density & \\\hline
 $\wfrac$ & Anchor BS (A-BS) fraction  & \\\hline
 $\PPP_u$, $\userdnsty$ & user PPP and density & $\userdnsty =1000$  per sq. km \\\hline
  $\PPP_\mu$, $\mu$ & UHF BS PPP and density & $\mu  =5$  per sq. km \\\hline
$\res$ &  mmWave  bandwidth & $2$ GHz \\  \hline
$\res_\mu$ &  UHF bandwidth & $20$ MHz \\\hline
$\power_{b}$ &  mmWave  BS transmit power & $30$ dBm \\\hline 
 $\power_{u}$ & user transmit power & $20$ dBm \\\hline
 $\sdpl$ & standard deviation of  path loss &  Access: LOS $ = 5.2$, NLOS = $7.6$\\
& &  Backhaul:  LOS  = $4.2$,  NLOS  = $7.9$\\\hline
 $\ple$ & path loss exponent &   Access: LOS $  = 2.0$, NLOS  $ = 3.3$\\
 & & Backhaul: LOS $ = 2.0$, NLOS $= 3.5$ \cite{GhoshJSAC} \\ \hline
  $\nu$ & mmWave carrier frequency &  $73 $ GHz  \\ \hline  
 $\fsldb$ &   path loss at $1$ m & $70$ dB\\ \hline
$\antgain_{\mathrm{max}}$,   $\antgain_{\mathrm{min}}$, $\bwidth$ &  main lobe gain, side lobe gain, beam-width & $\antgain_{\mathrm{max}} =18$ dB,   $\antgain_{\mathrm{min}}=-2$ dB, $\bwidth =10^o$ \\ \hline
$\plos, \dlos $ & fractional LOS  area $\plos$  in corresponding ball of  radius $\dlos$  & $0.11$, $200$ m  \\ \hline
$\sigma^2_N$ & noise power & $-174$ dBm/Hz + $10\log_{10}(\res)$ +  noise figure of  $10$ dB \\ \hline
\end{tabulary}
\end{table}

\begin{figure*}
\begin{equation}\label{eq:rate}
\rate = \begin{cases} \frac{\res}{\load{u,w}+\loadfrac\load{b}}\log(1+\SINR_{a}) & \hspace{-2in}\text{if associated with an A-BS,}\\
\frac{\res}{\load{u}}\min\left(\left(1-\frac{\loadfrac}{\loadfrac\load{b}+\load{u,w}}\right)\log(1+\SINR_{a}),\frac{\loadfrac}{\loadfrac\load{b}+\load{u,w}}\log(1+\SINR_{b})\right)
&\text{otherwise,}
\end{cases}
\end{equation}
\hrulefill
\vspace*{4pt}
\end{figure*}

\subsection{Blockage model}
 Each access  link of separation $d$ is assumed to be LOS with probability  $\plos$ if $d \leq \dlos$
and  $0$ otherwise\footnote{A fix LOS probability beyond distance $\dlos$ can also be handled as shown in Appendix \ref{sec:proofpldist}.}.  The parameter $\plos$ should be physically interpreted as the average fraction of LOS area in a circular ball of radius $\dlos$ around  the point under consideration. The proposed approach is simple yet flexible enough to capture  blockage statistics of real settings  as shown  in Sec. \ref{sec:validation}. The insights presented in this paper  corroborate those  from other blockage models too \cite{AkdRanRap13,GhoshJSAC,BaiHea14}.
The parameters ($\plos, \dlos$) are geography and deployment dependent (low for dense urban, high for semi-urban).  The analysis in this paper allows for different ($\plos$, $\dlos$) pairs for access and backhaul links.

\subsection{Association rule}  
Users are assumed to be associated (or served) by the BS offering the minimum path loss. Therefore, the BS serving  the user at origin is  $\sbs(0)\triangleq \arg\min_{X\in \PPP} \pl_a(X,0)$, where `a' (`b') is for access (backhaul).

 The index  $0$ is dropped henceforth wherever implicit.  The analysis in this paper is done for the user located at  the origin     referred to as the  \textit{typical} user\footnote{Notion of typicality is enabled by Slivnyak's theorem.} and its serving BS is the \textit{tagged} BS.   Further, each  BS (with no wired backhaul) is assumed to be backhauled over the air  to the A-BS offering the lowest path loss to it.  Thus, the A-BS (tagged A-BS) serving the tagged BS at $\sbs$ (if  not an A-BS itself)   is $ \swbs(\sbs) \triangleq  \arg\min_{Y\in \PPP_w} \pl_b(Y,\sbs)$, with $\sbs \notin \PPP_w$. This two-hop  setup is demonstrated in Fig. \ref{fig:TwoHopNet_Fig}. As a result, the access (downlink and  uplink), and backhaul link  $\SINR$ are
\begin{align*}
\SINR_d  &  = \frac{\power_b  \antgainmax  \pl_a(\sbs)^{-1}}{\intrfrnc_d +\noisepower}, \,\, \SINR_u   = \frac{\power_u  \antgainmax \pl_a(\sbs)^{-1}}{\intrfrnc_u +\noisepower},\\
\SINR_b  & = \frac{\power_b  \antgain^2_\mathrm{max} \pl_b(\sbs,\swbs)^{-1}}{\intrfrnc_b +\noisepower},
\end{align*}
respectively, 
where $\noisepower \triangleq \mathrm{N_0}\res$ is the thermal noise power and $\intrfrnc_{(.)}$ is the corresponding interference.

\begin{figure}
  \centering
\subfloat{\label{}\includegraphics[width= 0.8\columnwidth]{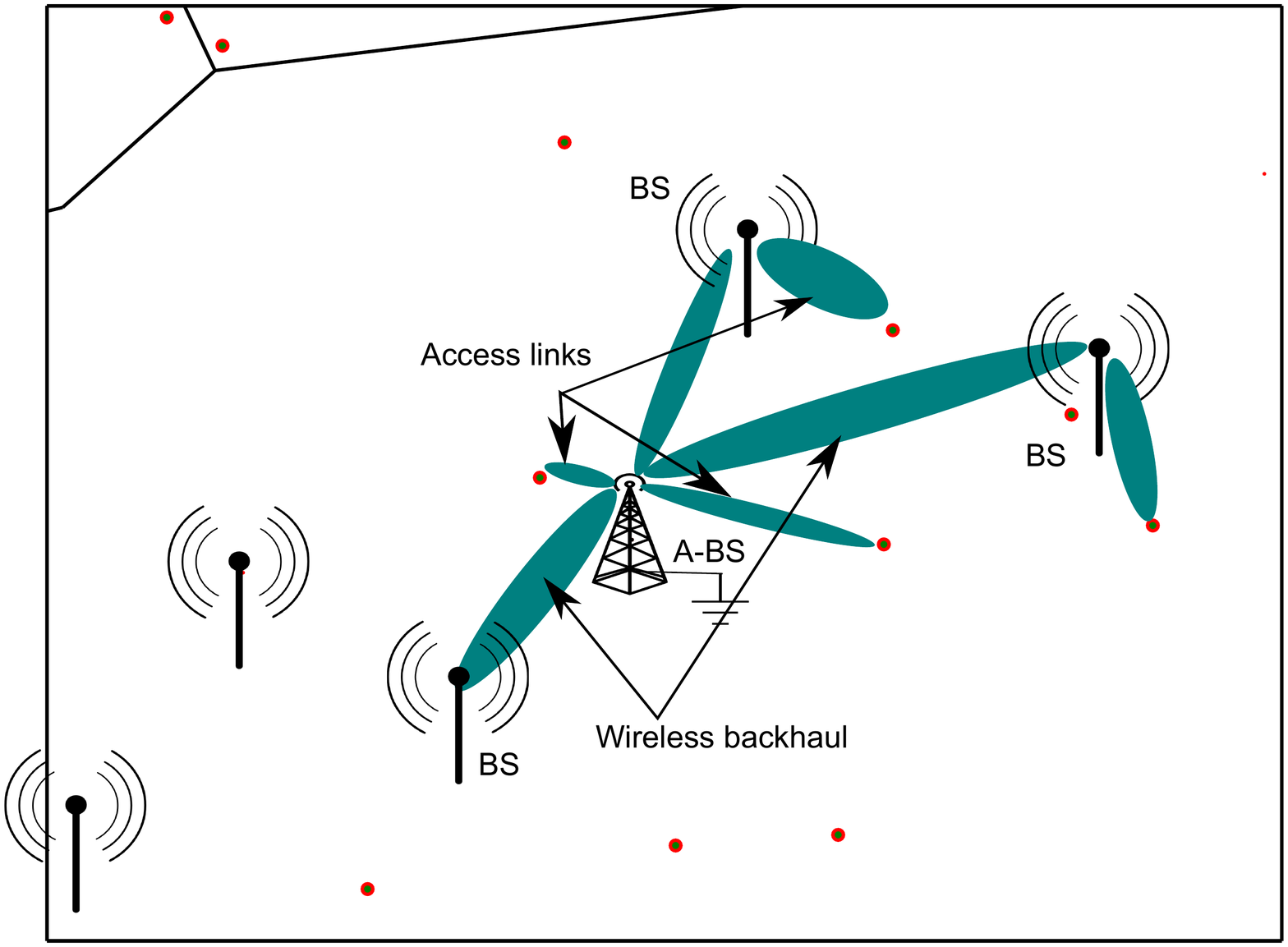}}
\caption{ Self-backhauled  network with the  A-BS providing the wireless backhaul to the associated BSs and access link to the associated  users  (denoted by   circles). }
 \label{fig:TwoHopNet_Fig}
\end{figure}

\subsection{Validation methodology}\label{sec:vmodel}
The analytical model and results   presented in this paper are validated using Monte Carlo simulations employing actual building topology     of two major metropolitan areas, Manhattan  and Chicago \cite{KulGC14}.   The polygons representing the buildings in the corresponding regions are  shown in Fig.~\ref{fig:urbanareas}. These regions represent  dense urban  settings, where mmWave networks are most attractive. 
In each simulation trial,  users and BSs are dropped  randomly in these geographical areas  as per the corresponding densities.   Users are dropped only in the outdoor regions, whereas the BSs landing inside a building polygon are assumed to be  NLOS to all users. A BS-user link is assumed to be NLOS if a building blocks the line segment joining the two, and LOS otherwise. The  association and propagation  rules are assumed as described in the earlier sections. The specific path loss parameters used are listed in Table \ref{tbl:param} and  are  from empirical measurements \cite{GhoshJSAC}. The association cells formed by two different  placements of mmWave BSs in downtown Manhattan  with this methodology  are shown in Fig. \ref{fig:ManAssoc}. 

\subsection{Access and backhaul load}
Access and backhaul links are assumed to share (through orthogonal division) the same pool of radio resources and hence the user rate depends on the user load at BSs and BS load at A-BSs. Let $\load{b}$, $\load{u,w}$, and $\load{u}$  denote the number of BSs associated with the tagged A-BS, number of users served by the tagged A-BS, and the number of users associated with the tagged BS respectively. By definition, when the typical user associates with an A-BS, $\load{u,w} = \load{u}$. 
Since an A-BS serves both users and BSs, the resources allocated to the associated BSs (which further serve their associated  users)  are assumed to be proportional to their average user load. Let the average number of users per BS be denoted by $\loadfrac \triangleq \userdnsty/\dnsty$, and  then the fraction of resources $\af_b$  available for all the associated BSs at an A-BS are  $\frac{\loadfrac \load{b}}{\loadfrac\load{b} + \load{u,w}}$, and those for the access link with the associated users are then $\af_{a,w} = 1-\af_b= \frac{\load{u,w}}{\loadfrac\load{b} + \load{u,w}}$. The fraction of resources reserved for the associated BSs at an A-BS are assumed to be shared equally among the BSs and hence the fraction of resources available to the tagged BS from the tagged A-BS are   $\af_b/\load{b}$, which is equivalent to the resource fraction  used for    backhaul   by the corresponding BS.   The access and backhaul capacity at each BS is assumed to be shared equally among the associated users. Furthermore, the rate of a user is assumed to equal the minimum of the access link rate and backhaul link rate.   

With the above described resource allocation model the  rate/throughput  of a     user is given by (\ref{eq:rate}) (at  top of the page), 
where $\SINR_a$ corresponds to the $\SINR$ of the access link: $a \equiv d$ for downlink and $a \equiv u$ for uplink.

\begin{figure*}
\begin{multline} \label{eq:intnsty}
\dnstyr_a((0,t])  = \dnsty\pi \plos\bigg\{\dlos^2\left[\Q\left(\frac{\ln(\dlos^{\ple_l}/t)-\meen_l}{\sigma_l}\right)-\Q\left(\frac{\ln(\dlos^{\ple_n}/t)-\meen_n}{\sigma_n}\right)\right]  +  t^{2/{\ple_{l}}}\exp\left(2\frac{\sigma_{l}^2}{\ple_{l}^2}+2 \frac{m_{l}}{\ple_{l}} \right)\\\times \Q\left(\frac{\sigma_{l}^2 (2/\ple_{l})-\ln(\dlos^{\ple_l}/t)+\meen_l}{\sigma_l}\right)
+   t^{2/{\ple_{n}}}\exp\left(2\frac{\sigma_{n}^2}{\ple_{n}^2}+2 \frac{m_{n}}{\ple_{n}} \right)\left[\frac{1}{\plos}-\Q\left(\frac{\sigma_{n}^2 (2/\ple_{n})-\ln(\dlos^{\ple_n}/t)+\meen_n}{\sigma_n}\right)\right]\bigg\} 
\end{multline}
\hrulefill
\vspace*{4pt}
\end{figure*}

\subsection{Hybrid networks}\label{sec:hybrid}
Co-existence with conventional UHF based $3$G and  $4$G  networks  could play a  key role in providing  wide coverage,  particularly in sparse deployment of  mmWave  networks, and reliable control channels. In this paper, a \textit{simple} offloading technique is adopted wherein a user is offloaded to the UHF network  if it's $\SINR$ on the mmWave network drops below a threshold $\tmin$. Since it was shown  in \cite{SinDhiAnd13} that the aggressiveness of offloading (or the offloading bias) is proportional to the bandwidth of the orthogonal band of small cells, the proposed $\SINR$-based association technique is arguably reasonable for  large bandwidth mmWave networks. A similar technique was also used in  \cite{LiIntel} for energy efficiency analysis.

\section{Rate Distribution: Downlink and Uplink}\label{sec:analysis}
This is the main technical section of the paper, which characterizes the   user rate distribution  across the network  in a self-backhauled mmWave network co-existing with  a  UHF macrocellular network.

\subsection{$\SNR$  distribution}
For characterizing the downlink  $\SNR$ distribution, the point process formed by the path loss of each BS to the typical user at origin defined as $\PropPP_a \coloneqq \left\{\pl_a(X)=\frac{\|X\|^{\ple}}{\sha}\right\}_{X \in \PPP} $, where $\sha \triangleq 10^{-(\plrnd+\fsldb)/10}$, on $\real{}$  is considered. Using the displacement theorem,  $\PropPP_a$ is a Poisson process and let the corresponding intensity measure be denoted by $\dnstyr_a(.)$.
\begin{lem}\label{lem:pldist}
The distribution of the path loss from the user to the tagged base station is such that $\pr(\pl_a(\sbs) > t) = \exp\left(-\dnstyr_a((0,t])\right)$, where the intensity measure is given by (\ref{eq:intnsty}) (at top of the page), 
where $m_{j} = -0.1\fsldb_j\ln10$, $\sigma_{j}=0.1 \sdpl_j\ln10$, with $j\equiv l$ for LOS and $j\equiv n$ for NLOS,  and $\Q(.)$ is the Q-function (Standard  Gaussian CCDF). 
\end{lem}
\begin{proof}
See Appendix \ref{sec:proofpldist}.
\end{proof}
The above lemma  simplifies to the scenario considered in  \cite{BlaKarKee12} with uniform path loss exponents (i.e. no blockage) and uniform shadowing variance. 

The path loss distribution for a typical backhaul link can be similarly obtained by considering the propagation process \cite{BlaKarKee12}  $\PropPP_b$  from A-BSs to the BS  at the origin. The corresponding  intensity measure $\dnstyr_b$ is then obtained by replacing $\dnsty$ by $\dnsty\wfrac$ and replacing the access link parameters with that of  backhaul link in (\ref{eq:intnsty}).

Under the assumptions of  stationary PPP for both users and BSs, considering  the typical link for analysis allows characterization  of the corresponding network-wide  performance metric.  Therefore, the $\SNR$ coverage defined as the distribution of $\SNR$ for the typical link $\scov_{(.)}(\SINRthresh)\triangleq \ppr_{\PPP_u}(\SNR_{(.)} > \SINRthresh)$ \footnote{$\ppr_\PPP$ is the Palm probability associated with the corresponding PPP $\PPP$. This notation is omitted henceforth with the implicit understanding that when considering the typical link, Palm probability is being referred to.} is also  
the complementary cumulative distribution function  (CCDF) of $\SNR$ across the entire network. The same holds for $\SINR$ and $\rate$ coverage. 

Lemma \ref{lem:pldist}  enables the characterization  of   $\SNR$ distribution in a closed form in the following theorem.
\begin{thm}\label{thm:pcov}
The  $\SNR$ distribution for the typical downlink, uplink, and backhaul link  are respectively
\begin{align*}
\scov_d(\SINRthresh) \triangleq \pr(\SNR_d > \SINRthresh)  & = 1 - \exp\left(- \dnsty \ndnstyr_a\left(\frac{\power_b\antgainmax}{\SINRthresh\noisepower}\right)\right)\\
\scov_u(\SINRthresh) \triangleq \pr(\SNR_u > \SINRthresh)  & = 1 - \exp\left(- \dnsty \ndnstyr_a\left(\frac{\power_u\antgainmax}{\SINRthresh\noisepower}\right)\right)  \\
\scov_b(\SINRthresh) \triangleq \pr(\SNR_b > \SINRthresh)  & = 1 - \exp\left(- \dnsty\wfrac\ndnstyr_b\left(\frac{\power_b\antgain^2_\mathrm{max}}{\SINRthresh\noisepower}\right)\right),
\end{align*}
where $\ndnstyr_a(t) \triangleq \frac{\dnstyr_a((0,t])}{\dnsty}$ and $\ndnstyr_b(t) \triangleq \frac{\dnstyr_b((0,t])}{\dnsty\wfrac}$.
\end{thm}
\begin{proof}
For the downlink case,
\begin{align*}
\pr(\SNR_d > \SINRthresh) &  = \pr \left(\frac{\power_b  \antgainmax  \pl_a(\sbs)^{-1}}{\noisepower} > \SINRthresh\right) \\&
 = 1 - \exp\left(-\dnsty \ndnstyr_a\left(\frac{\power_b \antgainmax}{\SINRthresh\noisepower}\right)\right),
\end{align*} where the last equality follows from Lemma \ref{lem:pldist}. Uplink and backhaul link coverage follow similarly.
\end{proof}
Noting the dependence of $\dnstyr_{(.)}(t)$ on $t$ and $\dnsty$, the $\SNR$ coverage (both access and backhaul) are directly proportional to the densities, power, and antenna gain of the respective links.  

As it can be noted,  users are assumed to be transmitting with maximum power in the uplink (without power control) in the above derivation.  This is arguably reasonable as the uplink $\SNR$ is already   problematic   in mmWave networks, even with max power transmission. However, the uplink $\SNR$ derivation  above can be extended to incorporate uplink fractional power control employed in   LTE networks, as shown in Appendix \ref{sec:proofulsinr}.

\subsection{Interference in mmWave networks}\label{sec:int}
This section provides an analytical treatment of interference in mmWave  networks. In particular, the focus of this section is to   upper bound  the interference-to-noise ($\INR$) distribution  (hence  provide  more insight into an earlier comment of noise-limited nature ($\SNR \approx \SINR$)  of  mmWave  networks), and quantify the impact of key design parameters on this upper bound.  Without any loss of generality, each BS is assumed to be an  A-BS (i.e. $\wfrac=1$) in this section and hence the subscript `a' for access is dropped. 

\begin{figure*}
\setcounter{MYtempeqncnt}{\value{equation}}
\setcounter{equation}{3}
\small
\begin{multline}\label{eq:mdt}
\ndnstyr'(t) \triangleq \frac{\mathrm{d}\ndnstyr(t)}{\mathrm{d}t}= \pi\plos\Bigg\{\frac{\dlos^2}{\sqrt{2\pi} t} \left[\frac{1}{\sigma_l}\exp\left(-\left(\frac{\ln (\dlos^{\ple_l}/t)-m_l}{\sqrt{2\sigma_{l}^2}}\right)^2\right)-\frac{1}{\sigma_n}\exp\left(-\left(\frac{\ln (\dlos^{\ple_n}/t)-\meen_n}{\sqrt{2\sigma_{n}^2}}\right)^2\right)\right]+\\
\exp\left(2\frac{\sigma_{l}^2}{\ple_{l}^2}+2 \frac{m_{l}}{\ple_{l}} \right) t^{\frac{2}{\ple_{l}}-1}\left[\frac{2}{\ple_l}\Q\left(\frac{\sigma_{l}^2 (2/\ple_{l})-\ln(\dlos^{\ple_l}/t)+\meen_l}{\sigma_l}\right) -\frac{1}{\sqrt{2\pi\sigma_l^2}}\exp\left(-\left(\frac{\sigma_{l}^2 (2/\ple_{l}) - \ln (\dlos^{\ple_l}/t)+m_l}{\sqrt{2\sigma_{l}^2}}\right)^2\right)\right]+\\
\exp\left(2\frac{\sigma_{n}^2}{\ple_{n}^2}+2 \frac{m_{n}}{\ple_{n}} \right) t^{\frac{2}{\ple_{n}}-1}\left[\frac{2}{\plos\ple_n}- \frac{2}{\ple_n} \Q\left(\frac{\sigma_{n}^2 (2/\ple_{n})-\ln(\dlos^{\ple_n}/t)+\meen_n}{\sigma_n}\right)+ \frac{1}{\sqrt{2\pi\sigma_n^2}}\exp\left(-\left(\frac{\sigma_{n}^2 (2/\ple_{n}) - \ln (\dlos^{\ple_n}/t)+\meen_n}{\sqrt{2\sigma_{n}^2}}\right)^2\right)\right]\Bigg\}.
\end{multline}\normalsize
\setcounter{equation}{\value{MYtempeqncnt}}
\hrulefill
\vspace*{4pt}
\end{figure*}
Consider the sum   over the earlier defined PPP $\PropPP$ 
\begin{equation}\label{eq:totI}
\totI \triangleq \sum_{Y \in \PropPP } Y ^{-1} K_Y,
\end{equation} 
where  $K_Y$ are  i.i.d. marks associated with  $Y \in \PropPP$. 
For example,  if $K_Y = \power_b \again_Y$ with $\again_Y$ being the random antenna gain on the link from $Y$,  then  $\totI$ denotes the  total received power from all BSs at the typical user.  The following proposition provides an upper bound to downlink  $\INR$  in mmWave networks.
\begin{prop}
The CCDF of  $\INR$   is upper bounded as 
\begin{equation*}
\pr(\INR>y) \leq  \frac{2e^{a\noisepower y}}{\pi}\int_{0}^\infty\mathrm{Re}(\bar{\mgf_{\totI}}(a+iu))\cos(
u\noisepower y) \mathrm{d} u, 
\end{equation*}
where $\bar{\mgf_{\totI}}(z) = 1/z- \mgf_{\totI}(z)/z$ with 
\begin{align*}
\mgf_{\totI}(z)= \exp\left(-\dnsty\cexpect{zK\int_{u>0}\frac{1-\exp(-u)}{u^2}\ndnstyr'\left(\frac{zK}{u}\right)\mathrm{d}u}{}\right)
\end{align*}
and $\ndnstyr'$ is given by (\ref{eq:mdt}) (at top of the page).  \addtocounter{equation}{1}
\end{prop}
\begin{IEEEproof}
 The downlink interference $\intrfrnc_d = \totI - K_{\sbs}/\sbs$ is clearly upper bounded by $\totI$  and hence $\INR \triangleq \intrfrnc_d /\noisepower$ has the property: $\pr(\INR>y) \leq \pr(\totI>\noisepower y)$.
The sum in (\ref{eq:totI}) is  the shot noise associated with   $\PropPP$ and the corresponding  Laplace transform is  represented as the Laplace functional of the shot noise of $\PropPP$,  $\mgf_{\totI}(z)\triangleq \expect{\exp(-z\totI)}$ 
\begin{equation*}
= \exp\left(-\cexpect{\int_{y>0}\left\{ 1-\exp(-zK/y)\right\}\dnstyr(\mathrm{d}y)}{K}\right),
\end{equation*} 
and the  Laplace transform associated with the CCDF of the shot noise is  $\bar{\mgf_{\totI}}(z) = 1/z- \mgf_{\totI}(z)/z$. 
The CCDF of the shot noise can then be obtained from the corresponding Laplace transform using the Euler characterization \cite{AbaWar95}
\begin{equation*}
\bar{F}_{\totI}(y)\triangleq \pr(\totI > y) =
\frac{2e^{ay}}{\pi}\int_{0}^\infty\mathrm{Re}(\bar{\mgf_{\totI}}(a+iu))\cos
uy \mathrm{d} u.
\end{equation*}
\end{IEEEproof}
The interference on the uplink   is generated by users transmitting on the same radio resource as the typical user. Assuming each  BS gives orthogonal resources to users associated with it,  one user per BS would interfere with the uplink transmission of the typical user. The point process of the interfering  users, for the analysis in this section, is assumed to be a PPP $\PPP_{u,b}$ of intensity same as that of BSs, i.e.,  $\dnsty$.  
In the same vein as the above discussion, the   propagation process $\PropPP_u \coloneqq \{\pl_a(X)\}_{X \in \PPP_{u,b}}$   captures the propagation loss from users to the BS under consideration at origin.  The  shot noise $\totI \triangleq \sum_{U \in \PropPP_{u} } U ^{-1} K_U$   then upper bounds the uplink interference with $K_U=  \power_u \again_U$.

 \begin{figure*}
  \centering
\subfloat[Manhattan]{\includegraphics[width= \columnwidth]{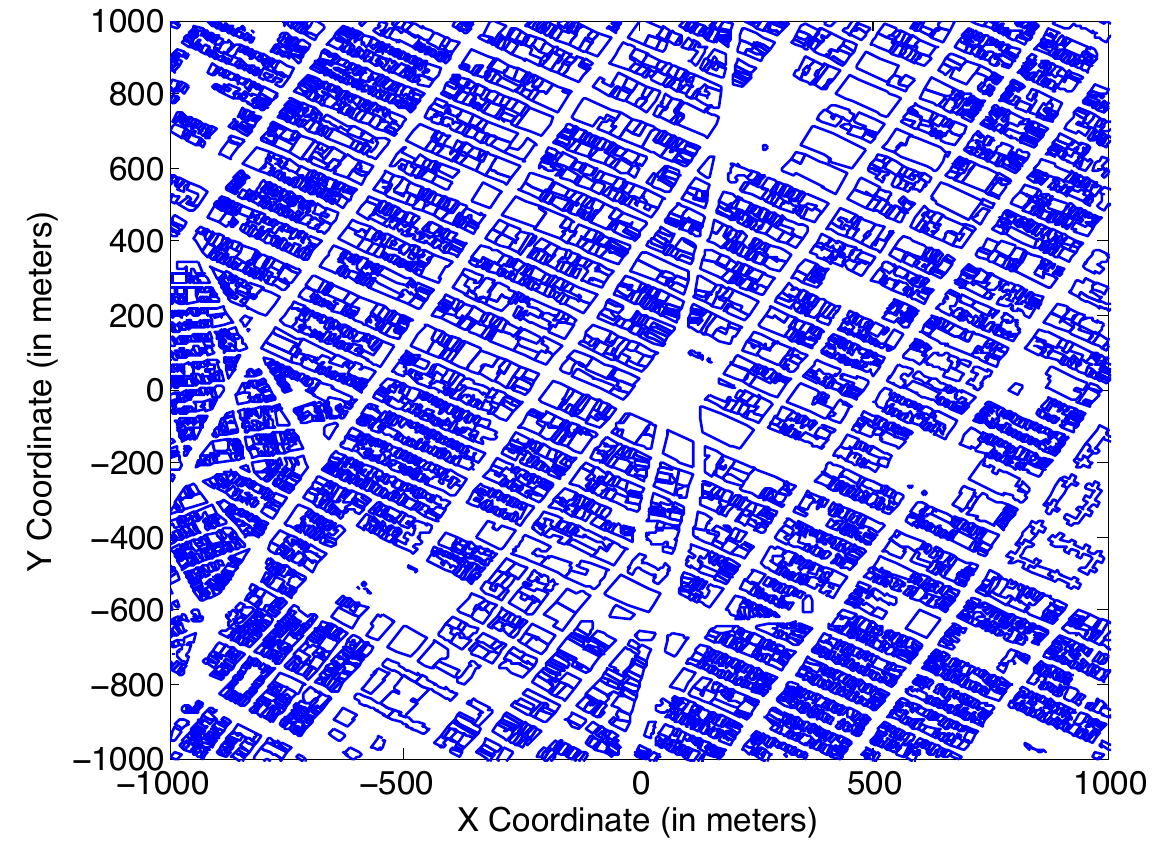}}
\subfloat[ Chicago]{\includegraphics[width= \columnwidth]{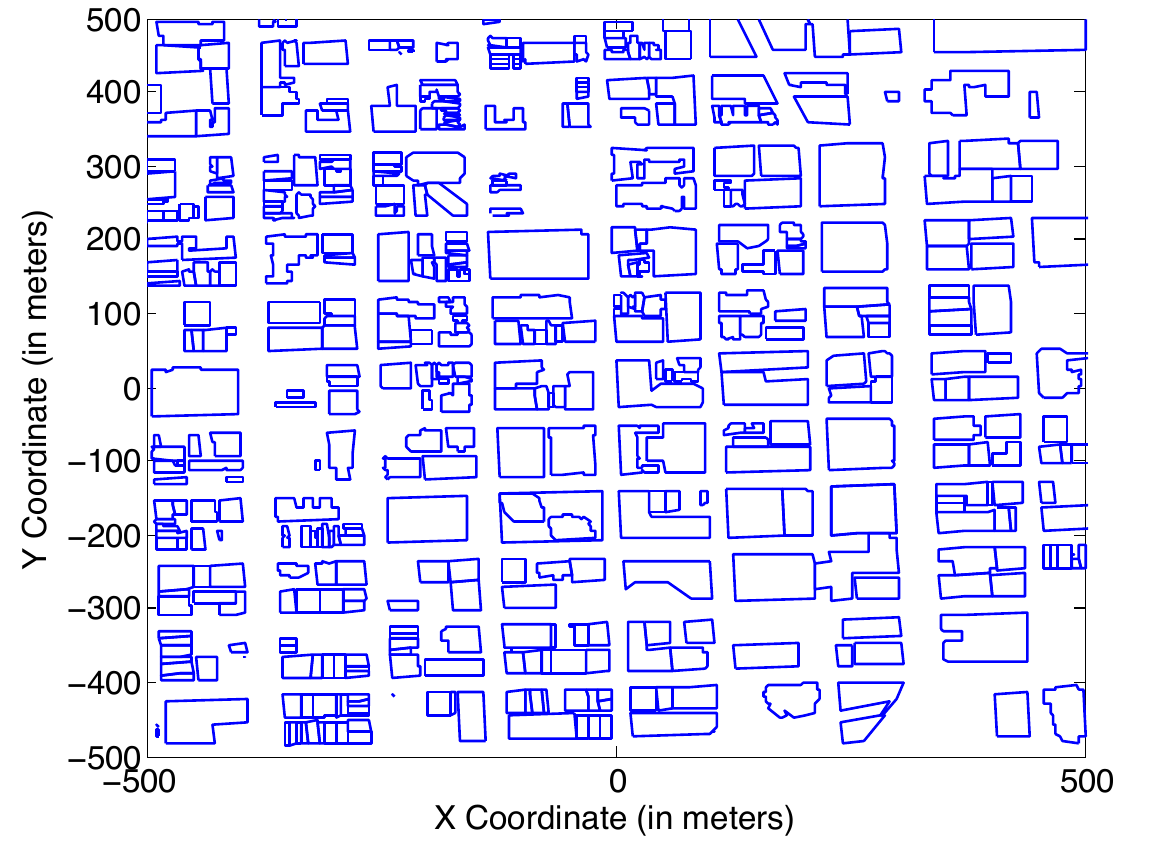}}
\caption{Building topology of    Manhattan   and Chicago  used for validation.}
 \label{fig:urbanareas}
\end{figure*} 

 \begin{figure*}
  \centering
\subfloat[]
{\label{ }\includegraphics[width= \columnwidth]{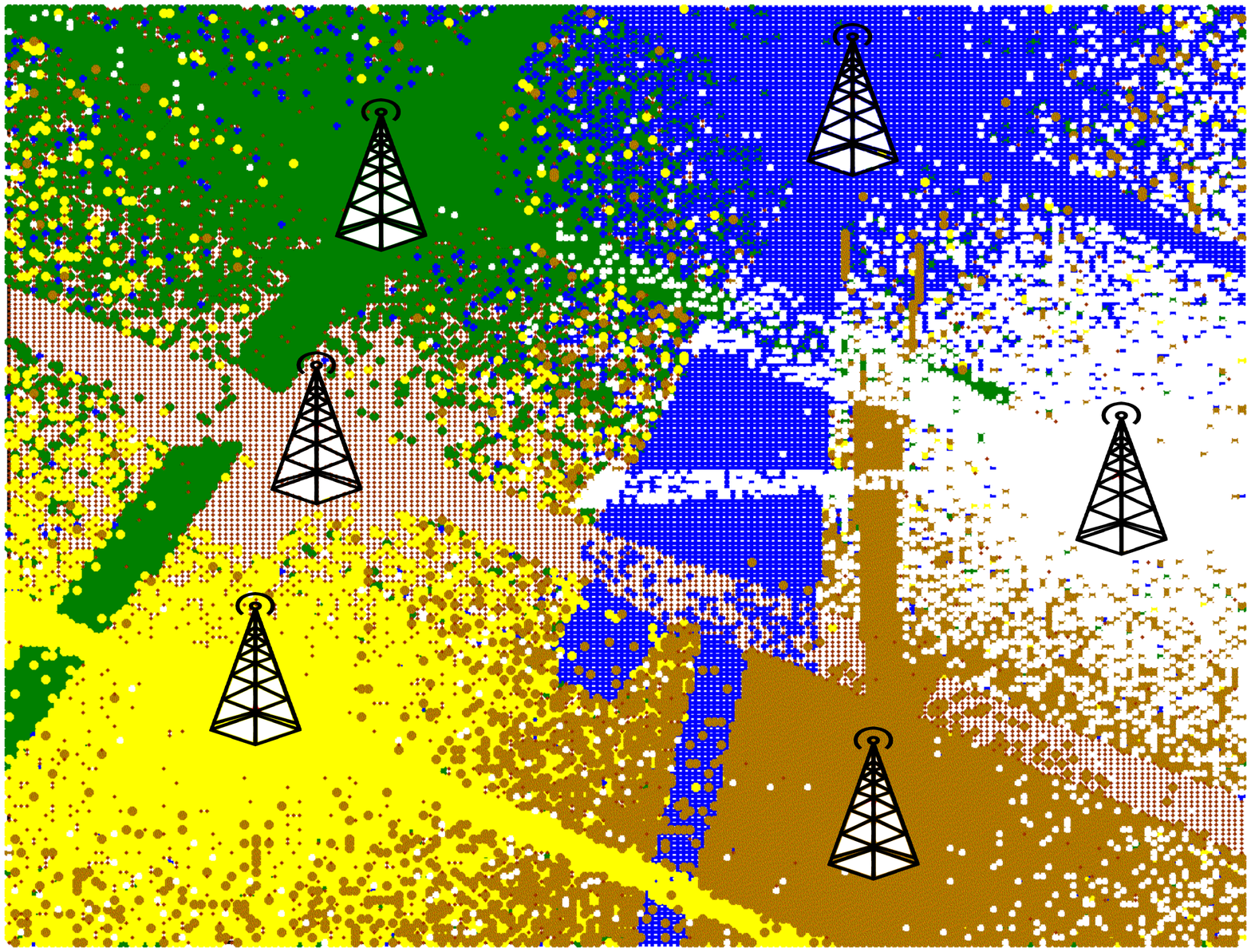}}
\subfloat[]{\label{ }\includegraphics[width =\columnwidth]{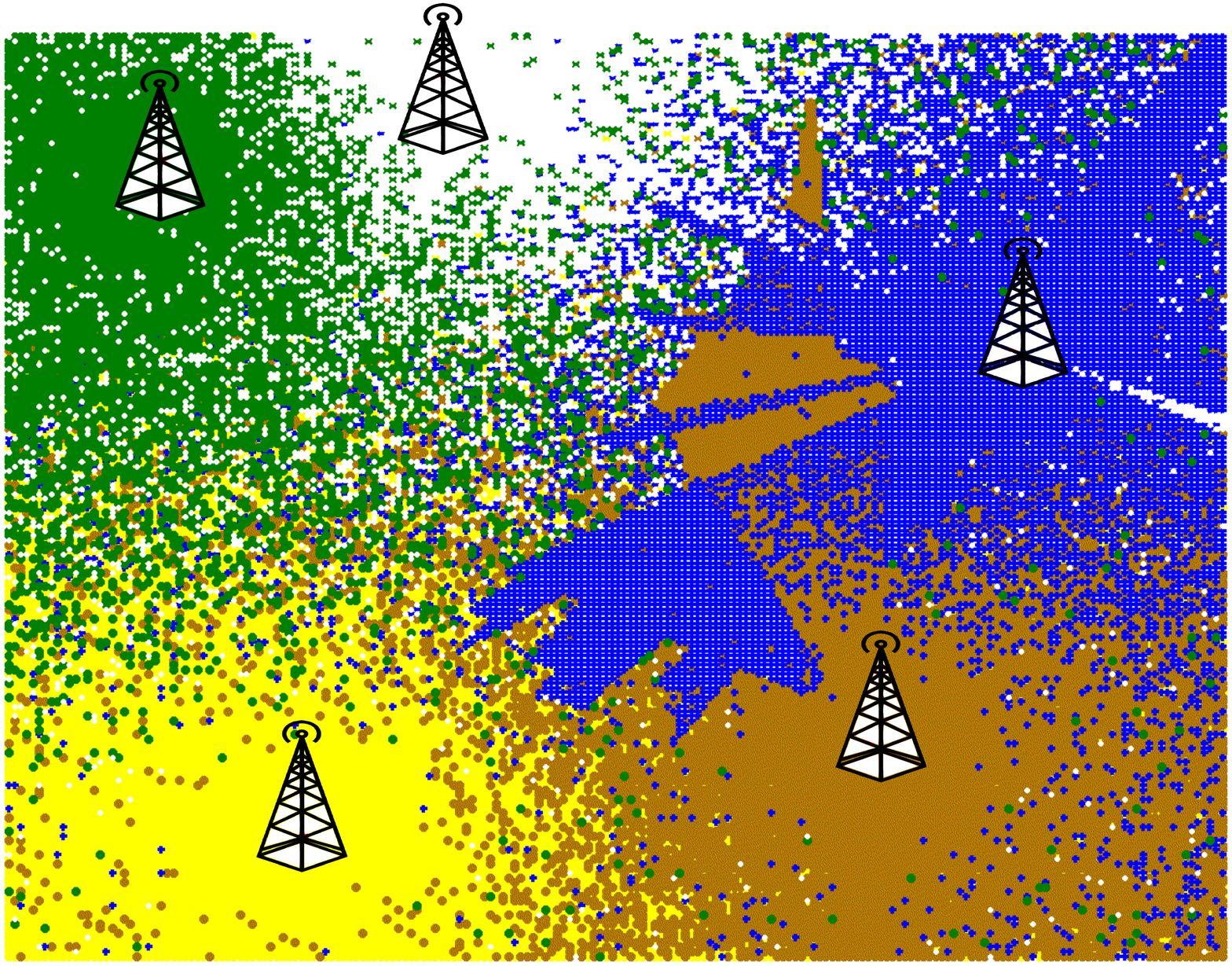}}
\caption{{Association cells  in different shades and colors for two different BS placement in Manhattan region. Noticeable discontinuity and irregularity of the cells show the     sensitivity of path loss to  blockages and the dense building topology (shown in Fig. \ref{fig:urbanareas}a).}}
 \label{fig:ManAssoc}
\end{figure*}

\begin{figure*}
  \centering
\subfloat[]
{\label{fig:INR}\includegraphics[width= \columnwidth]{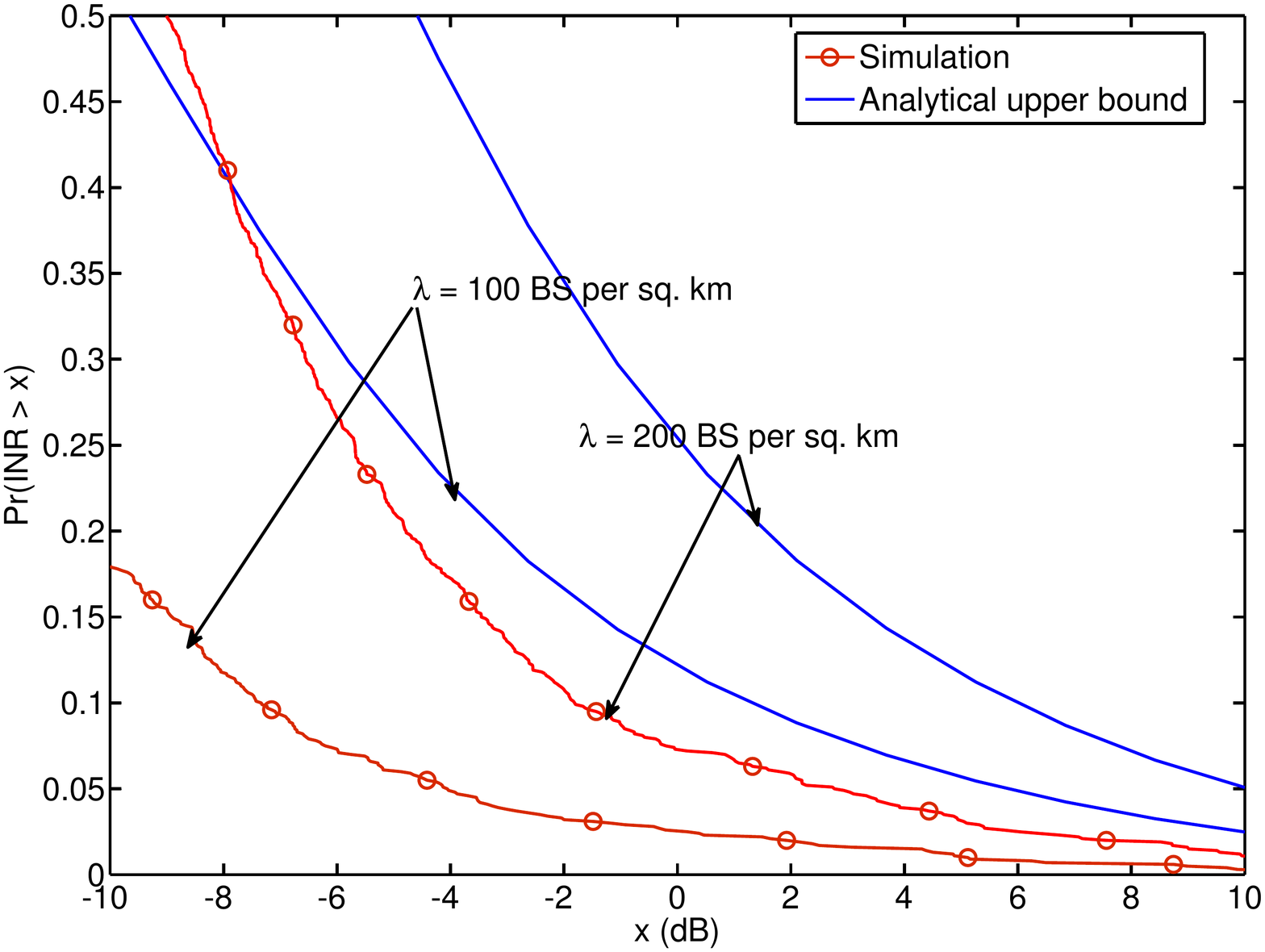}}
\subfloat[]{\label{fig:DensINRCov}\includegraphics[width= \columnwidth]{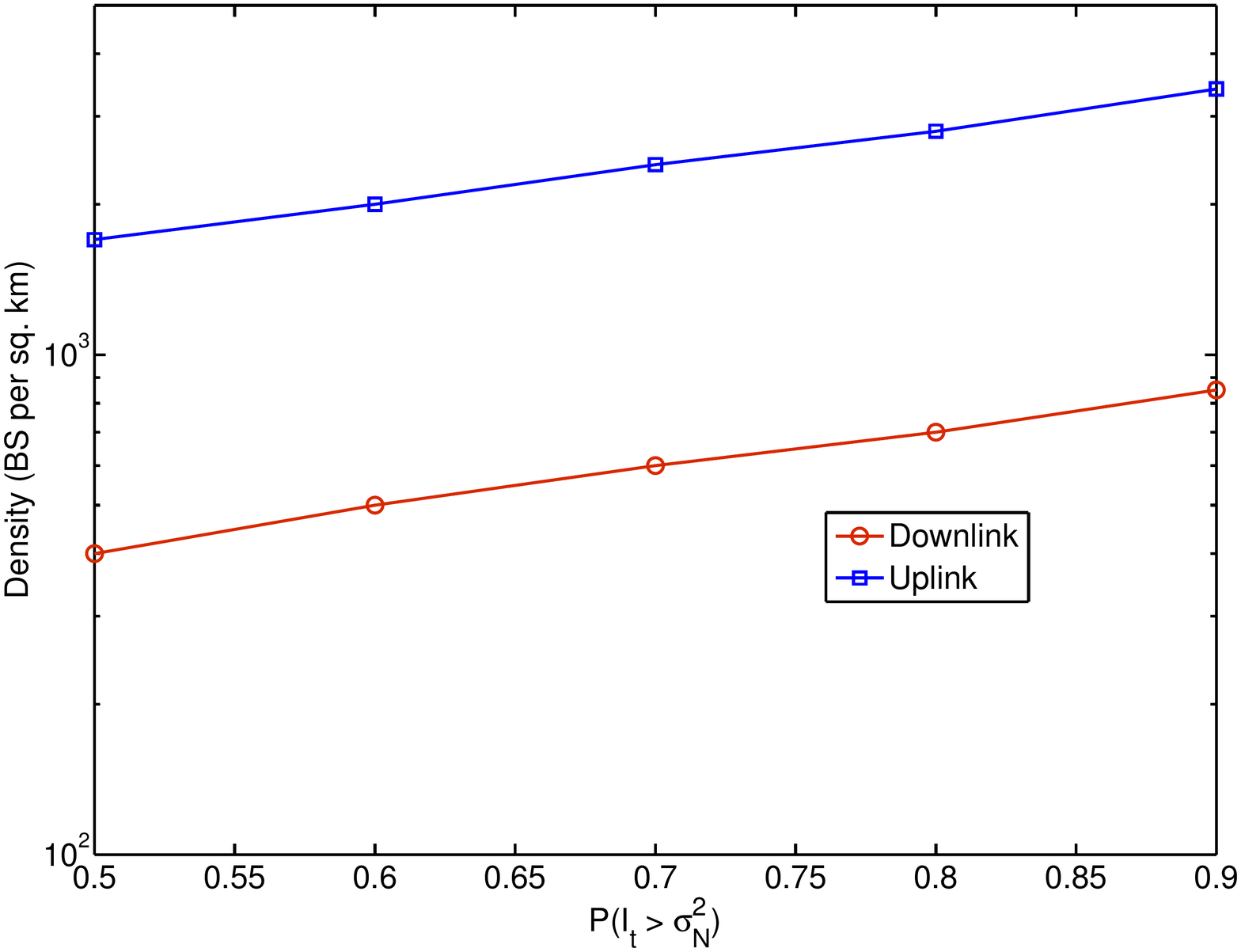}}
\caption{(a) Total power to noise ratio and $\INR$ for the proposed model, and (b) the variation of the   density required for the total power to exceed noise with a given probability.}
 \label{}
\end{figure*} 
The analytical total power to noise ratio bound for  the downlink with  the parameters of Table \ref{tbl:param} is shown in Fig. \ref{fig:INR}.  The Matlab code for computing the upper bound is available online \cite{Mcode}.  Also shown is the corresponding $\INR$ obtained through simulations. 
As can be observed from the analytical upper bounds and simulation, the interference  power does not dominate noise power for the large bandwidth and narrow beam-width network considered here. In fact, $\INR  > 0$ dB is observed in less than $20\%$ of the cases  even at  high base station densities of about $200$ per sq. km. As a consequence of  the   stochastic dominance,   the distribution of the total power (derived above) can be used to lower bound the density required for interference to dominate  noise. The minimum BS density required for achieving a given $\pr(\totI> \noisepower)$ for   uplink and downlink is shown in Fig. \ref{fig:DensINRCov}. As can be seen, a density of at least $500$ and $2000$ BS per sq. km is required  for guaranteeing downlink and uplink interference to exceed noise power with $0.7$ probability, respectively.  In general, the  $\INR$ distribution depends on the   bandwidth, antenna directivity (beam-width), carrier frequency, and density. The following corollary quantifies this effect. 
\begin{cor}\label{cor:equiv}\textbf{Density-Directivity-Bandwidth-Frequency Equivalence.}
In the case of  uniform path loss exponent ($\ple_l=\ple_n= \ple$)  and shadowing variance  for all links, the upper bound to the $\INR$ is proportional  to 
$\frac{\dnsty \power^{2/\ple}\expect{\again^{2/\ple}}}{\nu^{4/\ple}\res^{2/\ple}}.$
\end{cor}
\begin{proof}
 For the special case of uniform path loss exponent and shadowing variance  for all links,  
$\ndnstyr(u) = \pi\expect{\sha^{2/\ple}} u^{2/\ple}$ and $\ndnstyr'(u) = \frac{2\pi}{\ple}\expect{\sha^{2/\ple}}u^{2/\ple-1}$,   the Laplace transform  of $\totI$ is 
\begin{align*}
\mgf_{\totI}(z) &= \exp\left(2\pi\frac{\dnsty }{\ple}z^{2/\ple}\power^{2/\ple}\expect{\sha^{2/\ple}}\expect{\again^{2/\ple}}\Gamma\left(\frac{-2}{\ple}\right)\right),
\end{align*}
and the Laplace transform of $\totI/\noisepower$  is 
\[\exp\left(2\pi\frac{\dnsty }{\ple}\left(\frac{z}{\noisepower}\right)^{2/\ple}\power^{2/\ple}\expect{\sha^{2/\ple}}\expect{\again^{2/\ple}}\Gamma\left(\frac{-2}{\ple}\right)\right).\]
Noting the dependence of  thermal noise power $\noisepower$ on bandwidth and that of $\expect{\sha^{2/\ple}}$ on free space path loss (and thus on the carrier frequency) leads to the final result. 
\end{proof}
From the above corollary, it can be noted that the upper bound on the $\INR$ distribution is invariant with   increase in BS density or beam-width if the bandwidth and/or carrier frequency also scale appropriately.

The $\SINR$ distribution of the typical link defined as $\pcov_{(.)}(\SINRthresh)\triangleq \ppr_{\PPP_u}(\SINR_{(.)} > \SINRthresh)$ can be derived using the intensity measure of Lemma \ref{lem:pldist}  and is delegated to  Appendix \ref{sec:sinrproof}.  However, as   shown in this section, $\SNR$ provides a   good  approximation to   $\SINR$  for directional large bandwidth mmWave networks in densely blocked settings (typical for urban  settings), 
and hence the  following analysis   will,  \textit{deliberately}, ignore interference (i.e. $\pcov \approx \scov$).  However, the corresponding simulation results   include interference, thereby   validating this assumption.   For an  interference-limited setting, the analytical rate distribution results can  be obtained by replacing $\scov$ with $\pcov$.

\subsection{Load characterization}
As mentioned earlier,  throughput on  access and backhaul link depends on the number of users sharing the access link and the  number of BSs backhauling to the same A-BS respectively.  Hence there are two types of association cells in the network: \begin{inparaenum} \item user association cell of a BS -- the region in which all users are served by the corresponding BS, and \item BS  association cell of an A-BS -- the region in which all BSs are served by that A-BS. \end{inparaenum}
Formally, the user association cell of a BS (or an A-BS) located at $X \in \real{2}$ is 
\begin{equation*}
\assocr_X \triangleq  \left\{Y \in \real{2} : \pl_a(X,Y) < \pl_a(T,Y)\,\, \forall\,T \in \BSP\right\}
\end{equation*}
and the BS association cell of an A-BS located at $Z \in \real{2}$
\begin{equation*}
\assocr_{w,Z} \triangleq \left\{Y \in \real{2} :  \pl_b(Z,Y) < \pl_b(T,Y)\,\, \forall\,T \in  \WBSP\right\}. 
\end{equation*}
Due to the complex associations cells  in such networks, the resulting distribution of the association areas  (required for characterizing load distribution)  is  highly non-trivial to characterize exactly. The corresponding means, however, are characterized exactly  by the following remark.
\begin{rem}\textbf{Mean Association Areas.}
Under the modeling assumptions of  Sec. \ref{sec:sysmodel}, the minimum path loss association rule    corresponds to a \textit{stationary} (translation invariant) association \cite{SinBacAnd13}, and consequently the mean  user association area of a typical BS  equals the inverse of the corresponding density, i.e., $\mathbb{E}^o_{\BSP}\left[|\assocr_0|\right] = \frac{1}{\dnsty}$,  and  the mean BS association area  of a typical A-BS equals $\mathbb{E}^o_{\WBSP}\left[|\assocr_{w,0}|\right]=\frac{1}{\dnsty\wfrac}$.  Furthermore,   the area distribution of the tagged BS and A-BS follow an    area biased distribution as compared to that of the corresponding  typical areas resulting in the corresponding means to be $\dnsty \mathbb{E}^o_{\BSP}\left[|\assocr_0|^2\right]$ and $\dnsty\wfrac \mathbb{E}^o_{\WBSP}\left[|\assocr_{w,0}|^2\right]$ respectively.
\end{rem}
The above remark highlights that, although association regions are structurally very different from a distance-based Poisson-Voronoi (PV), they have the same mean areas as that of the PV with regards to the typical cell. This  leads to the next approximation.
\begin{asmptn} {\bf Association area distribution.}   The association area distribution of a typical BS and that of a typical A-BS is assumed to be same as that of the area distribution  of a typical PV with the same mean area (i.e. same density). 
\end{asmptn}
 The above approximation was proposed in \cite{SinDhiAnd13} for approximating area distribution of weighted PV and was verified through simulations. This approximation is validated in subsequent sections  using simulations in the context of rate distribution in mmWave networks. 
The probability mass function (PMF) of  the resulting loads based on the above discussion are stated below.  The proofs follow   along the  similar lines  of  \cite{SinDhiAnd13,YuKim13} and are thus omitted.
\begin{prop}\label{prop:loadprop} 
\begin{enumerate}
\item
The PMF of the number of users $\load{u}$ associated with the tagged BS  is 
\begin{equation*}
\loadpmf_t(\userdnsty,\dnsty,n) =\pr\left(\load{u}=n\right), 
\,\,\, \\n\geq 1,
\end{equation*}
where \[\loadpmf_t(c,d,n)= \frac{3.5^{3.5}}{(n-1)!}\frac{\Gamma(n+3.5)}{\Gamma(3.5)}\left(\frac{c}{d}\right)^{n-1}\left(3.5 + \frac{c}{d}\right)^{-(n+3.5)} ,\] and  $\Gamma(x)=\int_{0}^\infty \exp(-t)t^{x-1}\mathrm{d}t$ is the gamma function. The corresponding mean is $\avload{u}\triangleq \expect{\load{u}}=1+1.28\frac{\userdnsty}{\dnsty}$ \cite{SinDhiAnd13}.
 When the user associates with an A-BS $\load{u,w} = \load{u}$. Otherwise, the number of users $\load{u,w}$  served by the tagged A-BS   follow the same distribution as those in a typical BS given by 
\begin{equation*} 
\loadpmf(\userdnsty,\dnsty,n) =\pr\left(\load{u,w}=n\right),\,\,\,  \\n\geq 0,
\end{equation*}
where \[\loadpmf(c,d,n)= \frac{3.5^{3.5}}{n!}\frac{\Gamma(n+3.5)}{\Gamma(3.5)}\left(\frac{c}{d}\right)^{n}\left(3.5 + \frac{c}{d}\right)^{-(n+3.5)}. \] The corresponding mean is $\avload{u,w}\triangleq \expect{\load{u,w}}=\frac{\userdnsty}{\dnsty}$.
\item The number of BSs $\load{b}$ served by the tagged A-BS, when the typical user is served by the A-BS, has the same distribution as  the number of BSs associated with  a typical A-BS and hence
\begin{equation*} 
\loadpmf(\dnsty(1-\wfrac),\dnsty\wfrac,n) =\pr\left(\load{b}=n\right),\,\, \\n\geq 0.
\end{equation*} The corresponding mean is $\avload{b}\triangleq \expect{\load{b}}=\frac{1-\wfrac}{\wfrac}$.
 In the scenario where the typical user associates with a BS, the number of BSs $\load{b}$ associated with the tagged A-BS  is  given by 
\begin{equation*} 
\loadpmf_t(\dnsty(1-\wfrac),\wfrac\dnsty,n)=\pr\left(\load{b}=n\right),\,\,  \\n\geq 1. \nonumber
\end{equation*}  The corresponding mean is $\avload{b}=1+1.28\frac{1-\wfrac}{\wfrac}$.
\end{enumerate}
\end{prop}

\begin{figure*}
\begin{align}\label{eq:rcov}
& \rcov( \RATEthresh) \triangleq \pr(\rate > \RATEthresh)  = \wfrac\sum_{n\geq0,m\geq 1}\loadpmf(\dnsty(1-\wfrac),\dnsty\wfrac,n)\loadpmf_t(\userdnsty,\dnsty,m)\scov_d\left(\uRATEthresh\{\nRATEthresh(\loadfrac n +m)\}\right)\nonumber \\
& + (1-\wfrac)\sum_{l\geq1,n\geq1,m\geq0} \loadpmf_t(\userdnsty,\dnsty,l)\loadpmf_t(\dnsty(1-\wfrac),\wfrac\dnsty,n)\loadpmf(\userdnsty,\dnsty,m)\scov_b\left(\uRATEthresh\left\{\nRATEthresh l  (n+m/\loadfrac)\right\}\right)\scov_d\left(\uRATEthresh\left\{\nRATEthresh l \frac{n+m/\kappa}{n+m/\kappa-1}\right\}\right)
\end{align}
\hrulefill
\vspace*{4pt}
\end{figure*}

\begin{figure*}
\small
\begin{multline}\label{eq:rcovmean}
\bar{\rcov}( \RATEthresh)  = \wfrac\scov_d\left(\uRATEthresh\left\{\nRATEthresh\left(\frac{\userdnsty(1-\wfrac)}{\dnsty\wfrac} +1+1.28\frac{\userdnsty}{\dnsty}\right)\right\}\right)\\+ (1-\wfrac)\scov_b\left(\uRATEthresh\left\{\nRATEthresh\left(1+1.28\frac{\userdnsty}{\dnsty}\right)\left(2+1.28\frac{1-\wfrac}{\wfrac}\right)\right\}\right)\scov_d\left(\uRATEthresh\left\{\nRATEthresh \left(1+1.28\frac{\userdnsty}{\dnsty}\right) \frac{2+1.28(1-\wfrac)/\wfrac}{1+1.28(1-\wfrac)/\wfrac}\right\}\right)
\end{multline}\normalsize
\hrulefill
\vspace*{4pt}
\end{figure*}\normalsize

\subsection{Rate coverage}
As emphasized in the introduction, the rate distribution (capturing the impact of loads on access and backhaul links) is vital for assessing the   performance of self-backhauled mmWave networks. The   lemmas below characterize  the downlink rate distribution for a mmWave and a hybrid network employing the following approximations.  Corresponding results for the uplink are obtained by replacing $\scov_d$ with $\scov_u$.
\begin{asmptn}
The  number of users  $\load{u}$ served by the tagged BS  and the number of BSs $\load{b}$ served by the tagged A-BS       are assumed independent of each other and the corresponding link $\SINR$s/$\SNR$s. 
\end{asmptn}
\begin{asmptn}
The spectral efficiency of the tagged backhaul link is assumed to follow the same distribution as that of the typical backhaul link.
\end{asmptn}
\begin{lem}\label{lem:rcov}
The rate coverage of a typical user in a self backhauled mmWave  network, described in Sec. \ref{sec:sysmodel}, for a rate threshold $\RATEthresh$  is  given by (\ref{eq:rcov}) (at top of the page),  where $\nRATEthresh = \RATEthresh/\res$, $\uRATEthresh(x) =  2^x-1$, and $\scov_{(.)}$ are from Theorem \ref{thm:pcov}.
\end{lem}
\begin{proof}
Let $\passoc_w$ denote the event of the typical user associating with an A-BS, i.e.,  $\pr(\passoc_w) = \wfrac$. Then, using (\ref{eq:rate}), the rate coverage  is  
\footnotesize
\begin{align*}
&\rcov(\RATEthresh) =  \wfrac\pr\left( \frac{ \af_{a,w}}{\load{u,w}}\log(1+\SINR_{d})> \nRATEthresh | \passoc_w\right)+ (1-\wfrac) \times \\&\pr\left(\frac{1}{\load{u}}\min\left(\left(1-\frac{\af_b}{\load{b}}\right)\log(1+\SINR_{d}),\frac{\af_b}{\load{b}}\log(1+\SINR_{b})\right) > \nRATEthresh | \bar{\passoc_w}\right)\\
& = \wfrac\expect{\scov_d(\uRATEthresh(\nRATEthresh\left\{\load{u,w}+\loadfrac\load{b}\right\}))} + (1-\wfrac)\times\\&  \expect{\scov_d\left(\uRATEthresh\left\{\nRATEthresh \load{u}\frac{\load{b}+\load{u,w}/\loadfrac}{\load{b}+\load{u,w}/\loadfrac-1}\right\}\right) \scov_b(\uRATEthresh\{\nRATEthresh\load{u}(\load{b}+\load{u,w}/\loadfrac)\})}.
\end{align*}
\normalsize
The rate coverage expression then follows by invoking the independence among various loads and $\SNR$s. 
\end{proof}
In case the  different loads in the above lemma are  approximated with their respective means, the rate coverage expression is simplified as in the following corollary.
\begin{cor}\label{cor:meanrcov}
The rate coverage with mean load approximation using Proposition \ref{prop:loadprop} is given by (\ref{eq:rcovmean}) (at top of the page). 
\end{cor}

As can be observed from the above corollary, increasing the fraction of A-BSs $\wfrac$ in the network increases the  probability of being served by  an A-BS (the weight of the first term). The rate from an A-BS ($\scov_d(.)$ in the first term)  also increases with  $\wfrac$, as user and BS load per A-BS decreases. Furthermore, increasing $\wfrac$ also increases the backhaul rate  ($\scov_b(.)$ in the second term) of a user associated with a BS. Further investigation into the interplay of $\dnsty$, $\wfrac$, and rate is deferred to Sec. \ref{sec:sb}.

\begin{rem} In practical communications systems, it might be unfeasible to transmit reliably with any   modulation and coding (MCS) below a certain $\SNR$: $\SINRthresh_0$ (say), and in that case $\rate=0$ for $\SNR< \SINRthresh_0$. Such a constraint can be incorporated in the above analysis by replacing   $\uRATEthresh  \to \max(\uRATEthresh,\SINRthresh_0)$.
\end{rem}

The following lemma  characterizes the rate distribution in a hybrid network with the association technique of Sec. \ref{sec:hybrid}.
\begin{lem}
The rate distribution in a hybrid  mmWave network (with $\wfrac=1$) co-existing with a  UHF macrocellular network, described in Sec. \ref{sec:hybrid}, is  
\begin{multline*}
 \rcov_H( \RATEthresh) = \rcov_{1}(\RATEthresh)  + (1-\scov_d(\tmin))\\\times \sum_{n\geq1}\loadpmf_t(\userdnsty-\lambda_{u,m},\mu,n)\pcov_\mu(\uRATEthresh\left\{\RATEthresh n/\res_\mu\right\})\nonumber,
 \end{multline*}
where  $\rcov_{1}(\RATEthresh)$ is obtained from Lemma \ref{lem:rcov} by replacing $\userdnsty\to\lambda_{u,m}\triangleq\userdnsty\scov_d(\tmin)$ (the  effective density of users associated with mmWave network) and $\uRATEthresh \to \uRATEthresh_1 \triangleq\max(\uRATEthresh, \tmin)$, $\pcov_\mu$ is the $\SINR$ coverage on UHF network, and $\loadpmf_t(\userdnsty-\lambda_{u,m},\mu,n)$ is the PMF of the number of users $\load{\mu}$ associated with the tagged UHF BS.
\end{lem}
\begin{proof}
Under the association method of Sec. \ref{sec:hybrid}, the rate coverage in the hybrid setting is
\begin{align*}
\pr(\rate > \RATEthresh)  & = \pr(\rate  > \rho \cap \SINR_d > \tmin) \\& + \pr(\rate  > \rho \cap \SINR_d < \tmin )\\
&= \rcov_1(\RATEthresh) + (1-\scov_d(\tmin)) \expect{\pcov_\mu(\uRATEthresh\left\{\RATEthresh/\res_\mu\load{\mu}\right\})},
\end{align*}
where the first term on the RHS  is the rate coverage when associated with the mmWave network and hence  $\rcov_1$ follows from the previous Lemma \ref{lem:rcov}  by incorporating the offloading $\SINR$ threshold and reducing the user density to account for the users offloaded to the macrocellular   network (fraction $1-\scov_d(\tmin)$). The second term is the rate coverage when associated with the UHF network and $\load{\mu}$ is the load on the tagged UHF BS, whose distribution can be expressed as in \cite{SinDhiAnd13}  noting the mean association cell area of a UHF BS is $\frac{1-\scov_d(\tmin)}{\mu}$. The UHF network's $\SINR$ coverage $\pcov_\mu$ can be derived as in earlier work \cite{BlaKarKee12,andganbac11}.
\end{proof}

\subsection{Validation}\label{sec:validation}

\begin{figure}
 	\centering
 		\includegraphics[width=\columnwidth]{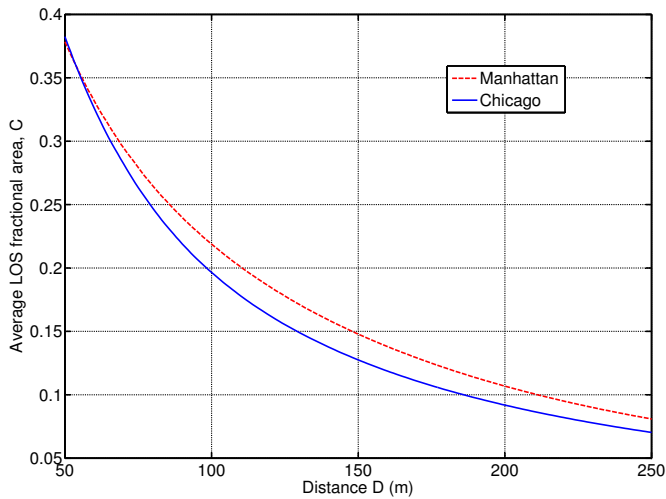}
 		\caption{Average LOS fractional area as a function of radius $\dlos$  averaged over the respective geographical regions.}
 	\label{fig:alos}
 \end{figure}

In the proposed model, the primary geography dependent parameters are  $\plos$ and $\dlos$.    As mentioned earlier, for a given $\dlos$,  the parameter $\plos$   is the average LOS fractional area in a disk of radius  $\dlos$. In order to fit the proposed model to a particular geographical region, the following methodology is adopted. Using Monte Carlo simulations in the setup of Sec. \ref{sec:vmodel}, the average fraction of LOS area  in a disk of radius $\dlos$ around randomly dropped  users   is obtained as a function of the radius $\dlos$. Fig.~\ref{fig:alos} shows the empirical $\plos$  obtained by averaging over  the  Manhattan and Chicago regions of Fig. \ref{fig:urbanareas}.  
 \begin{table}
 \begin{center}
      \caption{Values of  $\dlos$ and $ \plos$}
     \begin{tabular}{ | c | c | c |}
     \hline
     Urban area  & $\mathrm{D}$  (m) &  $\plos$\\\hline
     Chicago &  250 & 0.07 \\ \hline
     Manhattan & 200 & 0.11 \\ \hline
     \end{tabular}
     \label{tab:fitpara}
 \end{center}
 \end{table}
\begin{figure*}
  \centering
\subfloat[Manhattan]
{\includegraphics[width=\columnwidth]{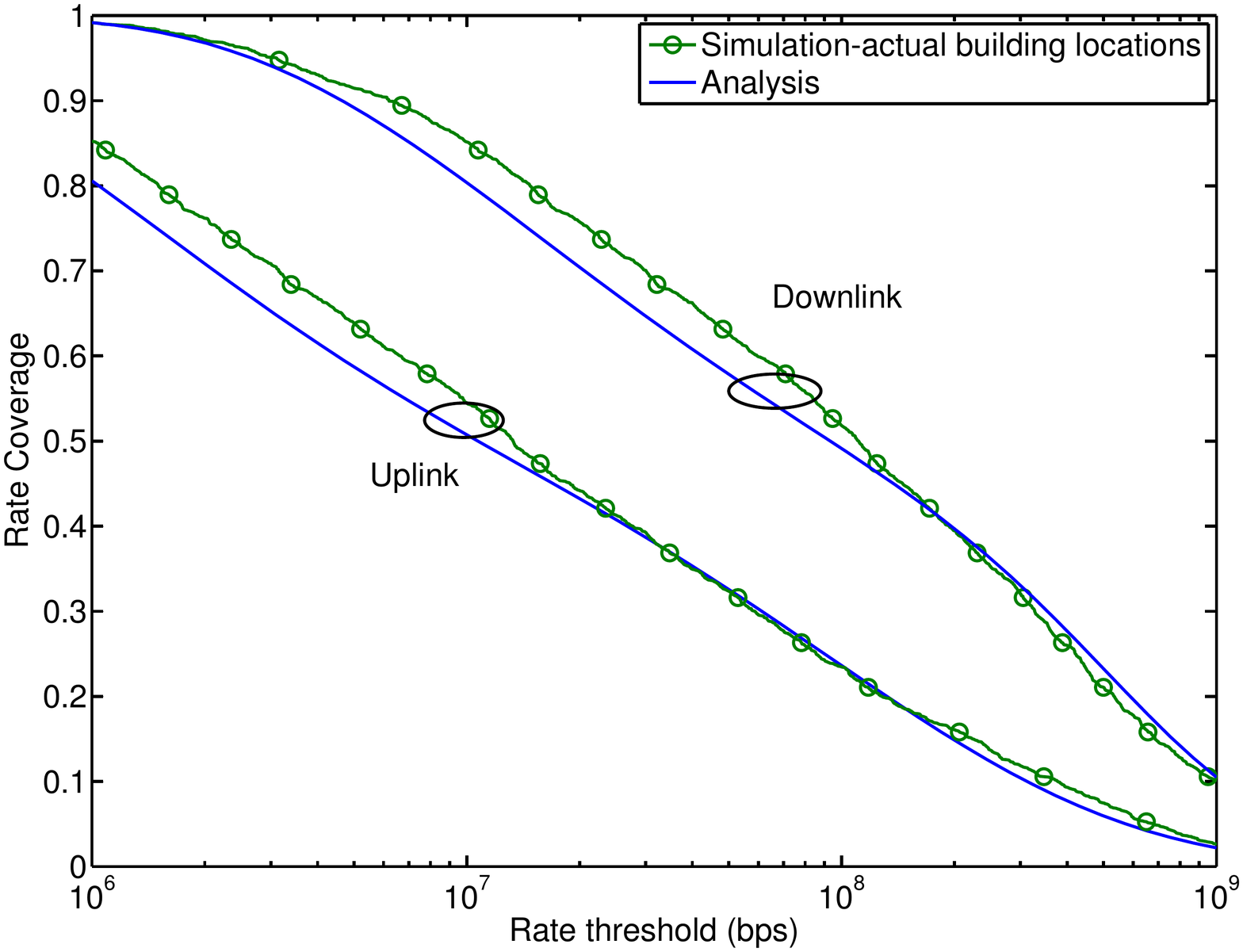}}
\subfloat[Chicago]{\includegraphics[width=\columnwidth]{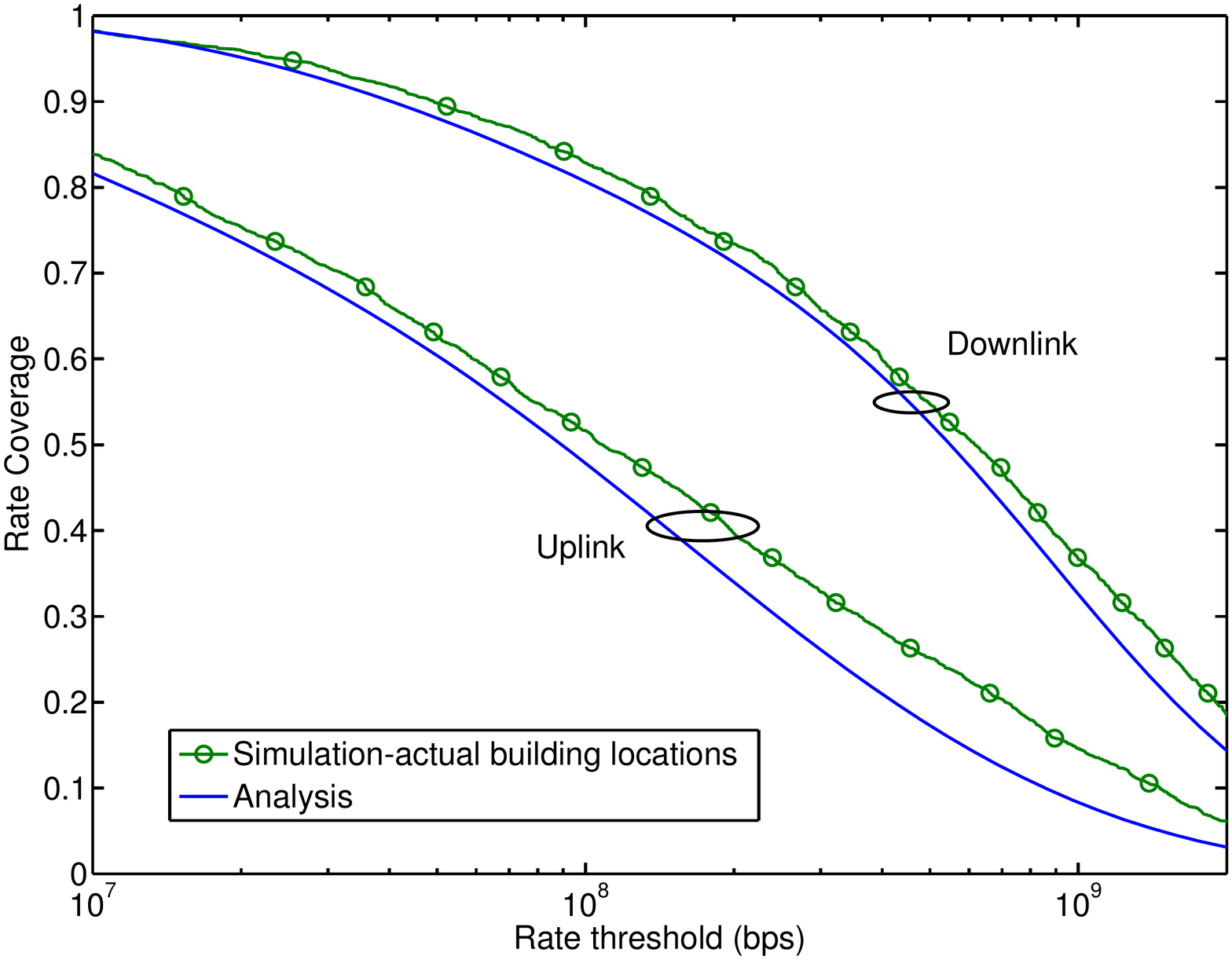}}
\caption{Downlink   rate distribution comparison from  simulation and analysis for   BS density $30$ per sq. km  in Manhattan (a) and $60$ per sq. km in   Chicago (b) with user density of $200$ per sq. km.}
 \label{fig:RateFit}
\end{figure*}
The downlink rate distribution (both uplink and downlink) obtained from  simulations (as per Sec. \ref{sec:vmodel}) and analysis (Lemma \ref{lem:rcov}) is shown in  Fig.~\ref{fig:RateFit} for the two cities with two different different BS densities and user density of $200$ per sq. km.  The parameters ($\plos$, $\dlos$) used in analysis for the specific geography are obtained    using  Fig.~\ref{fig:alos} and are given in  Table~\ref{tab:fitpara}.  The closeness of the analytical results to those  of the simulations validates \begin{inparaenum}[(a)] \item the ability of the proposed simple blockage model to capture the blockage characteristics of   dense urban settings, and \item the load characterization for irregular association cells (Fig.~\ref{fig:ManAssoc}) in a mmWave network. \end{inparaenum} The  closeness  of the match  builds confidence in the model and the  derived design insights.

In the above plots any   ($\plos$, $\dlos$) pair from Fig. \ref{fig:alos} can be used. However, it is observed that the match is better for the ($\plos$, $\dlos$) pair with larger $\dlos$ ($200$-$250$m, see \cite{KulGC14} for robustness analysis). This is due to the fact that   the LOS fractional area ($\plos_{\bar{\dlos}}$, say) beyond distance $\dlos$ is ignored, which is a better approximation for larger $\dlos$. It is straightforward to allow LOS area outside $\dlos$ in the analysis (as shown in Appendix \ref{sec:proofpldist}) but estimating  the same using actual building locations  is quite computationally intensive and tricky, as averaging needs to be done over a considerably larger area.  The fit procedure is simplified,  though not sacrificing the accuracy of the fit much (as seen), by setting $\plos_{\bar{\dlos}}=0$ in the model. 
 
\begin{figure*}
  \centering
\subfloat[Downlink]
{\label{}\includegraphics[width=\columnwidth]{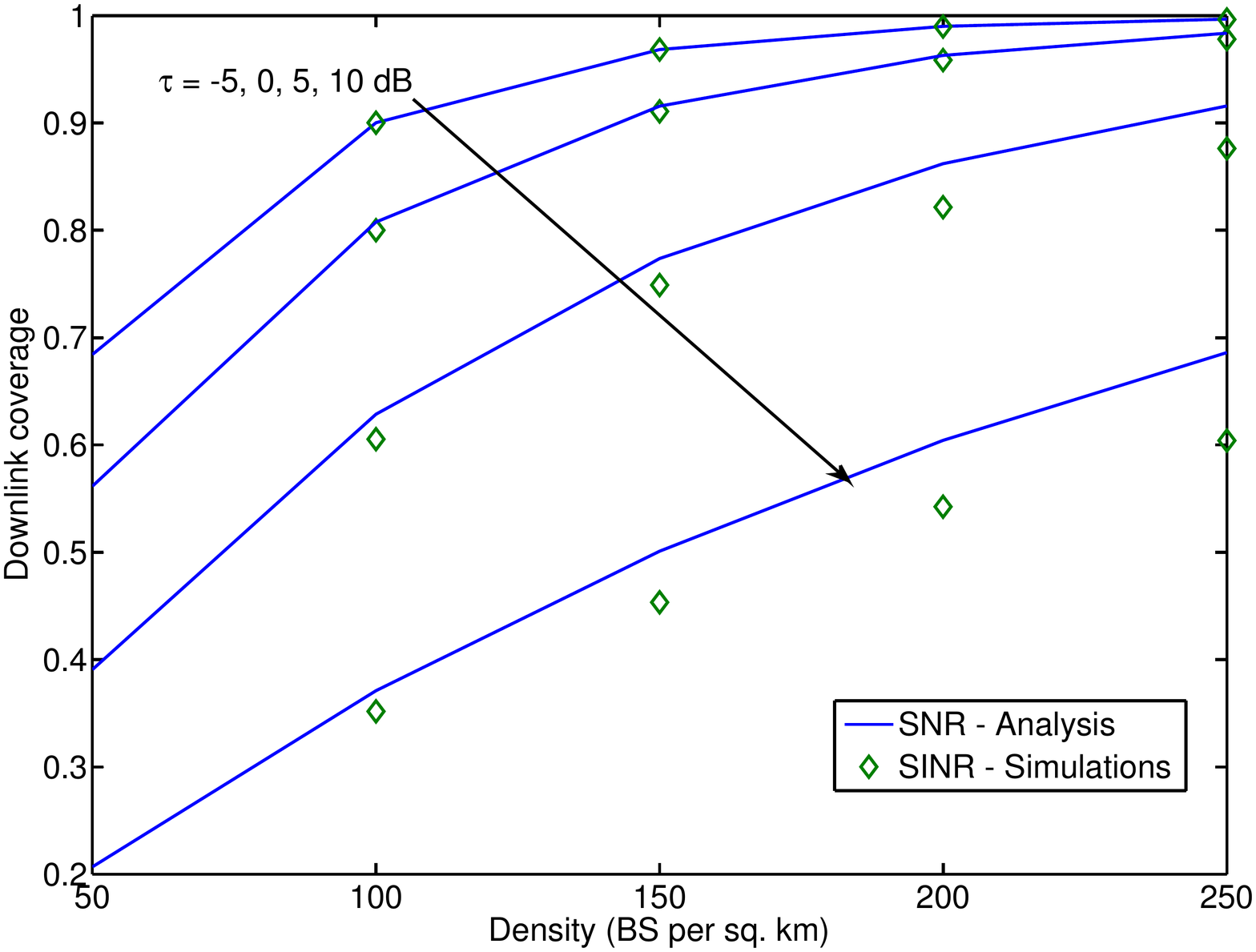}}
\subfloat[ ]{\label{}\includegraphics[width=\columnwidth]{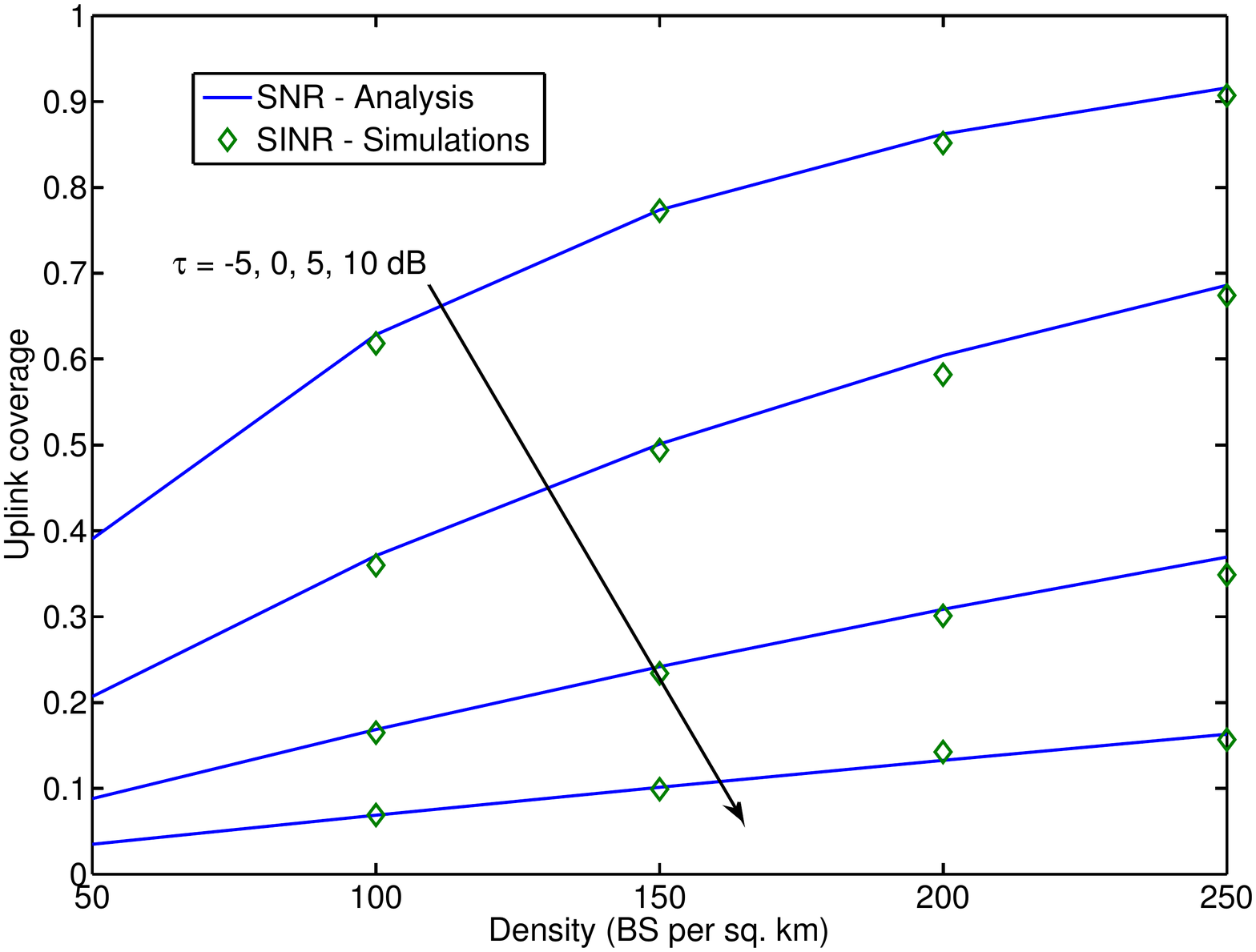}}
\caption{Comparison of  $\SINR$ (analysis) and $\SINR$ (simulation) coverage  with varying  BS density.}
 \label{fig:CovDens}
\end{figure*}

\section{Performance analysis and trends}\label{sec:results}

\subsection{Coverage and density }
The downlink and uplink   coverage for  various thresholds and density of BSs is  shown in Fig. \ref{fig:CovDens}. There are two  major observations:
\begin{itemize}
\item The analytical $\SNR$ tracks the  $\SINR$ obtained from simulation quite well for both downlink and uplink. A small gap ($< 10\%$) is observed for an example downlink case  with larger BS density ($250$ per sq. km)   and a higher threshold of $10$ dB.
\item Increasing the BS density
improves both the downlink and uplink coverage and hence the spectral efficiency -- a trend  in contrast to   conventional interference-limited networks, which are nearly  invariant in $\SINR$ to density.
\end{itemize}
As seen in Sec. \ref{sec:int}, interference is expected to dominate the thermal noise  for very large densities. The trend for downlink $\SINR$ coverage (derived in Appendix \ref{sec:sinrproof} assuming exponential fading power gain) for such  densities is shown in Fig. \ref{fig:SINRDens} for lightly  ($\plos=0.5$) and densely blocked ($\plos=0.11$) scenarios.   All BSs are assumed to be transmitting in Fig. \ref{fig:SINRdensa}, whereas BSs only with a user in the corresponding association cell are assumed to be transmitting in Fig. \ref{fig:SINRdensb}. The coverage for the latter case is obtained by thinning the interference field by   probability $1-\loadpmf(\userdnsty,\dnsty,0)$ (details in Appendix \ref{sec:sinrproof}). As can be seen, ignoring the finite user population, the $\SINR$ coverage saturates, where that saturation is achieved quickly for lightly blocked scenarios--a trend corroborated by the observations of \cite{BaiHea14}.  However, accounting for the  finite user population leads to a  \textit{different} trend, as the increasing density monotonically improves the path loss to the tagged BS, but the interference is (implicitly) capped by the finite user density of  $1000$ per sq. km.

\begin{figure*}
  \centering
\subfloat[All BSs transmit]
{\label{fig:SINRdensa}\includegraphics[width= \columnwidth]{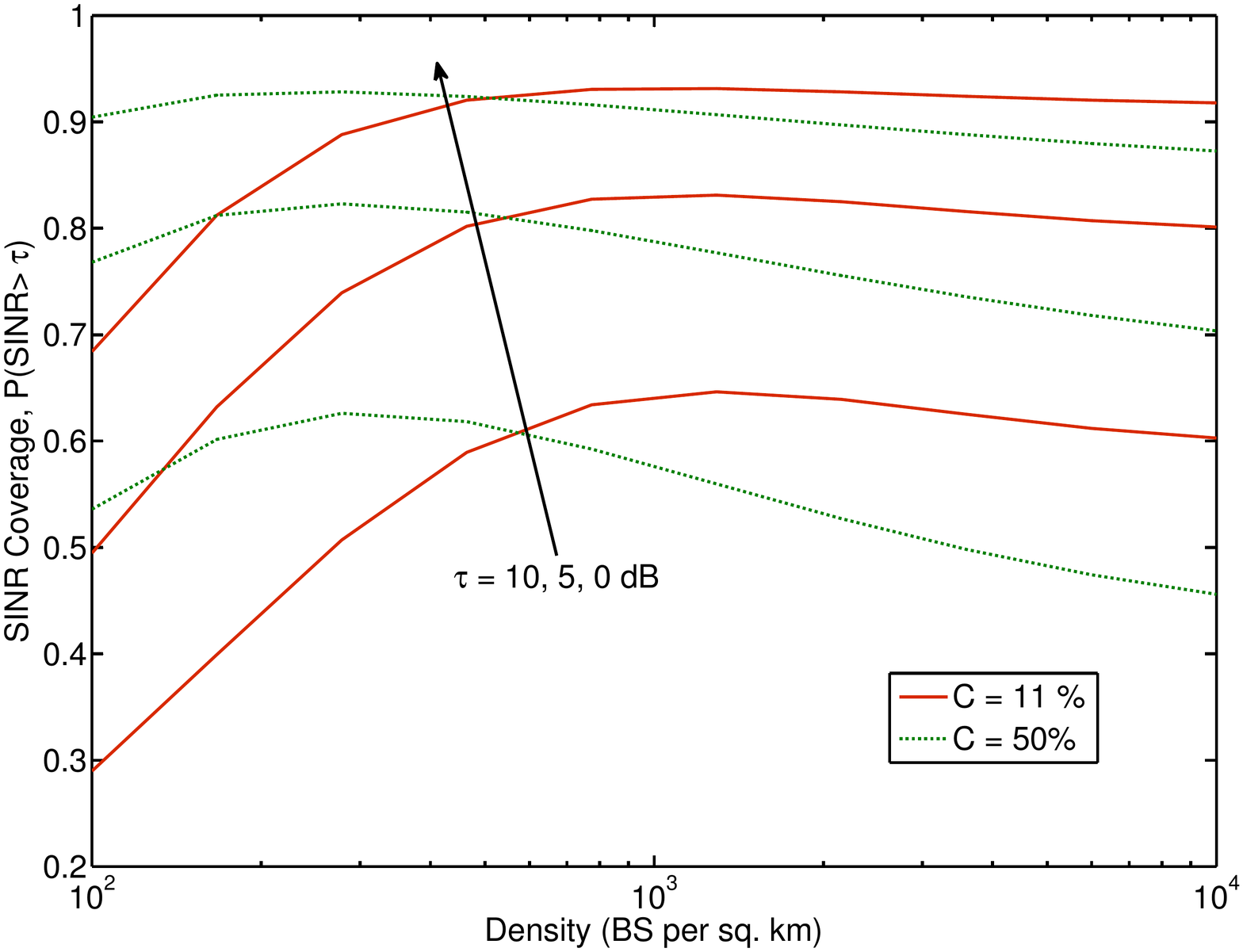}}
\subfloat[BSs with an active user transmit]
{\label{fig:SINRdensb}\includegraphics[width=\columnwidth]{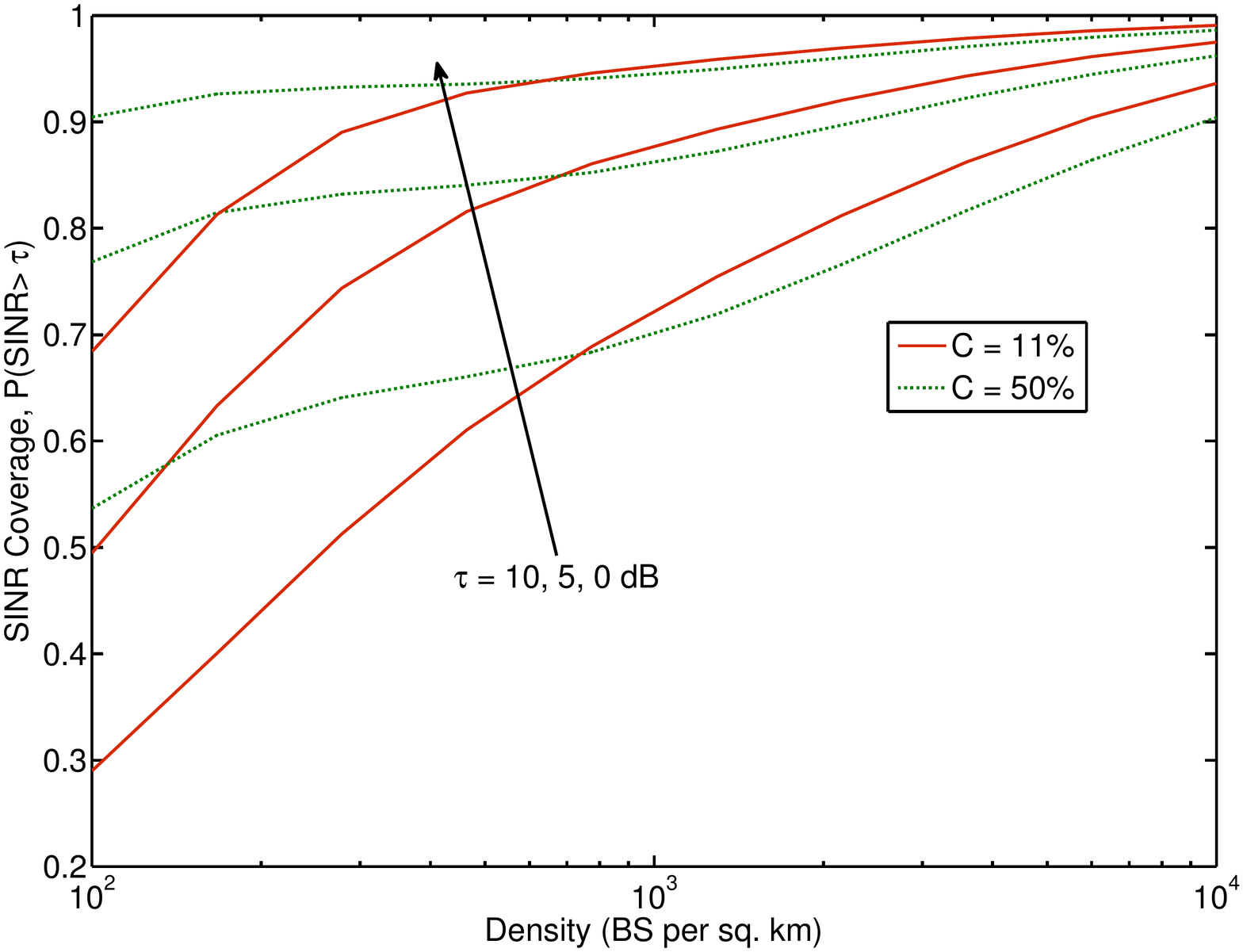}}
\caption{$\SINR$ coverage variation with large densities for different blockage densities.}
 \label{fig:SINRDens}
\end{figure*}

\subsection{Rate coverage}
The variation of downlink and uplink rate distribution with  the  density of infrastructure for a fixed A-BS fraction $\wfrac = 0.5$  is shown in Fig \ref{fig:RateCovDens}. Reducing the cell size by increasing density boosts   the  coverage  and   decreases the load per base station. This dual benefit  improves the overall rate drastically with density as shown in the plot. Further, the good match  of analytical curves to that of simulation also validates the analysis for  uplink and downlink rate coverage.
\begin{figure*}
  \centering
\subfloat[Downlink ]
{\includegraphics[width= \columnwidth]{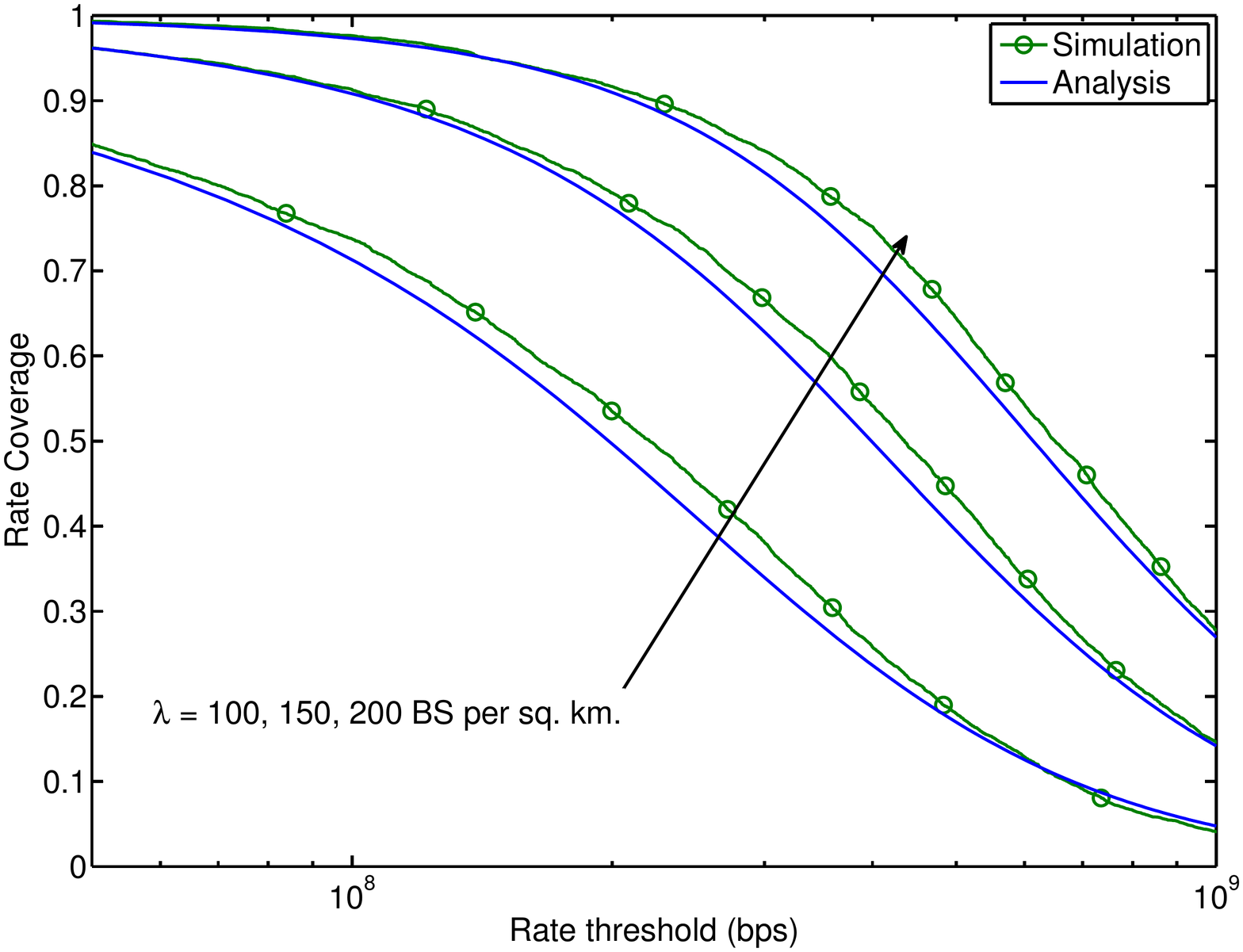}}
\subfloat[Uplink]{\includegraphics[width= \columnwidth]{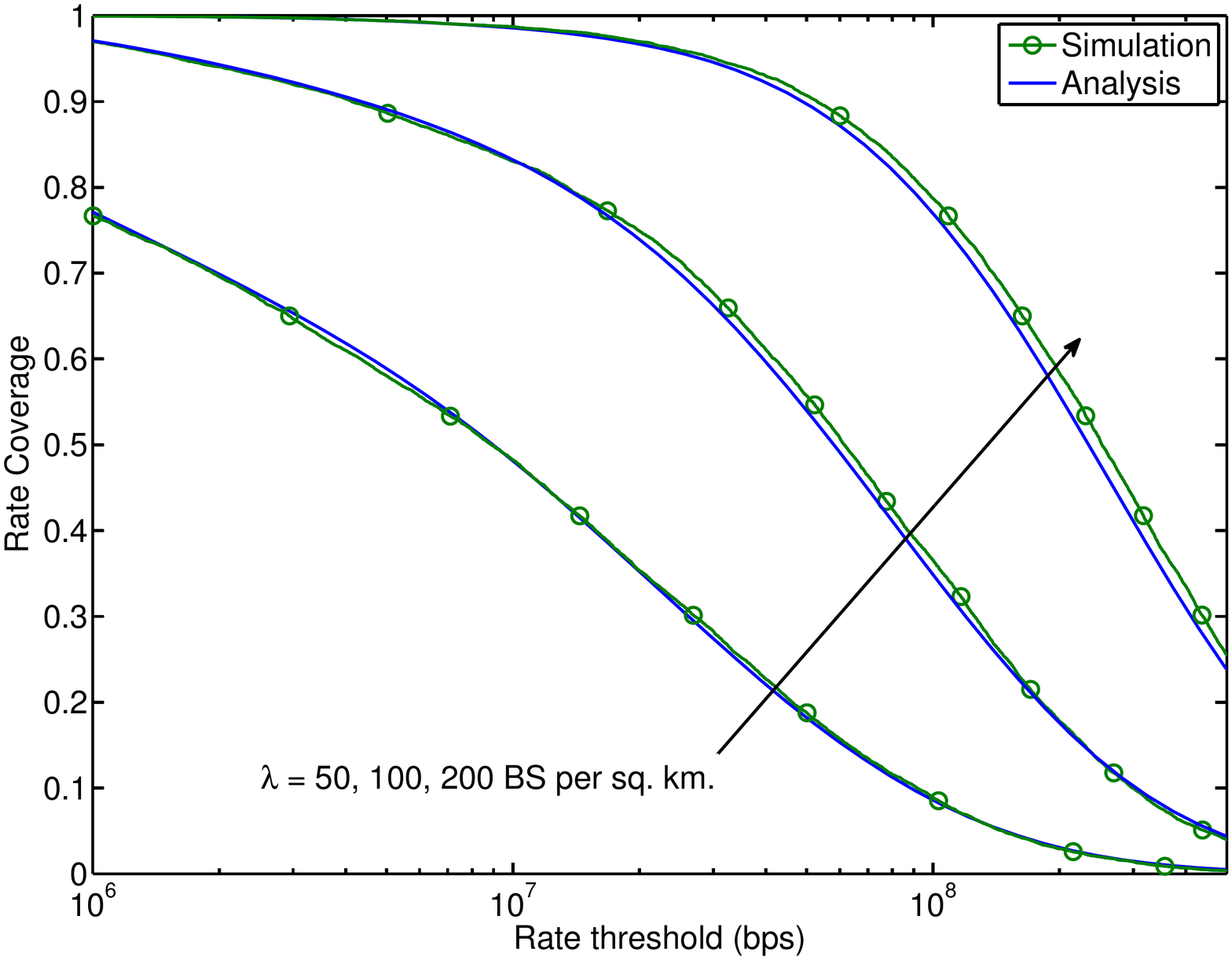}}
\caption{ Downlink and uplink rate coverage for different BS densities and fixed $\wfrac = 0.5$.}
 \label{fig:RateCovDens}
\end{figure*} 

The variation in  rate distribution with bandwidth is shown in Fig. \ref{fig:RateCovBW} for a fixed BS density $\dnsty =100$ BS per sq. km and $\wfrac=1$. Two observations   can be made here: \begin{inparaenum}[1)] \item   median and peak rate increase considerably with the availability of larger bandwidth,   whereas   \item cell edge rates  exhibit  a non-increasing  trend. \end{inparaenum} The latter trend  is due to the low $\SNR$ of the cell edge   users, where the gain from bandwidth is counterbalanced  by the loss in $\SNR$.  Further, if the constraint of $\rate = 0$  for $\SNR< \SINRthresh_0$ is imposed,   cell edge rates would actually decrease as shown in Fig. \ref{fig:RateCovBWMinConstraint} due to the increase in $\pr(\SNR<\SINRthresh_0)$, highlighting the impossibility of  increasing   rates for power-limited  users in  mmWave networks by just increasing the system bandwidth. In fact, it may be counterproductive.
\begin{figure*}
  \centering
\subfloat[$\SINRthresh_0 =0$]
{\includegraphics[width=\columnwidth]{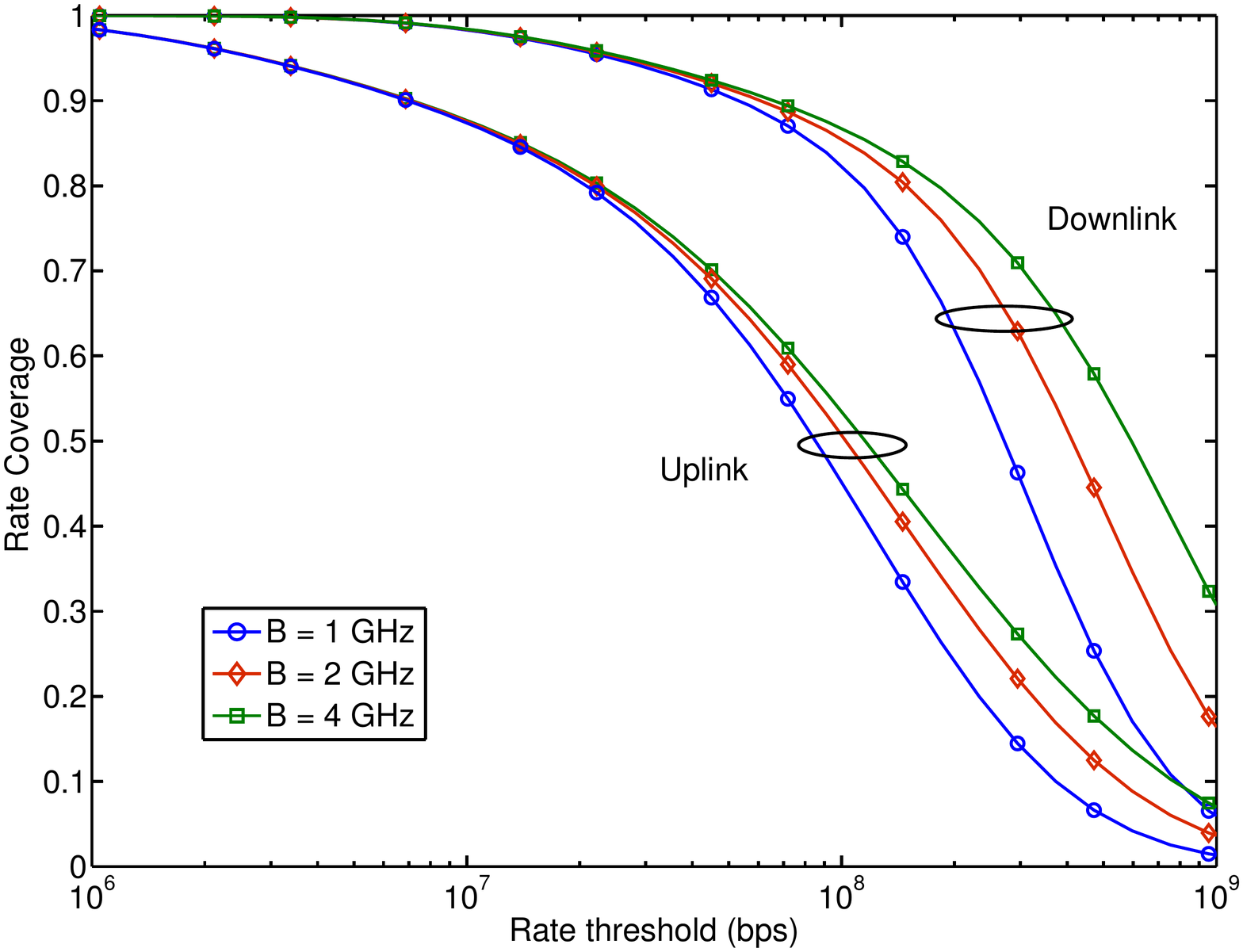}}
\subfloat[$\SINRthresh_0 =0.1$]{\label{fig:RateCovBWMinConstraint}\includegraphics[width=\columnwidth]{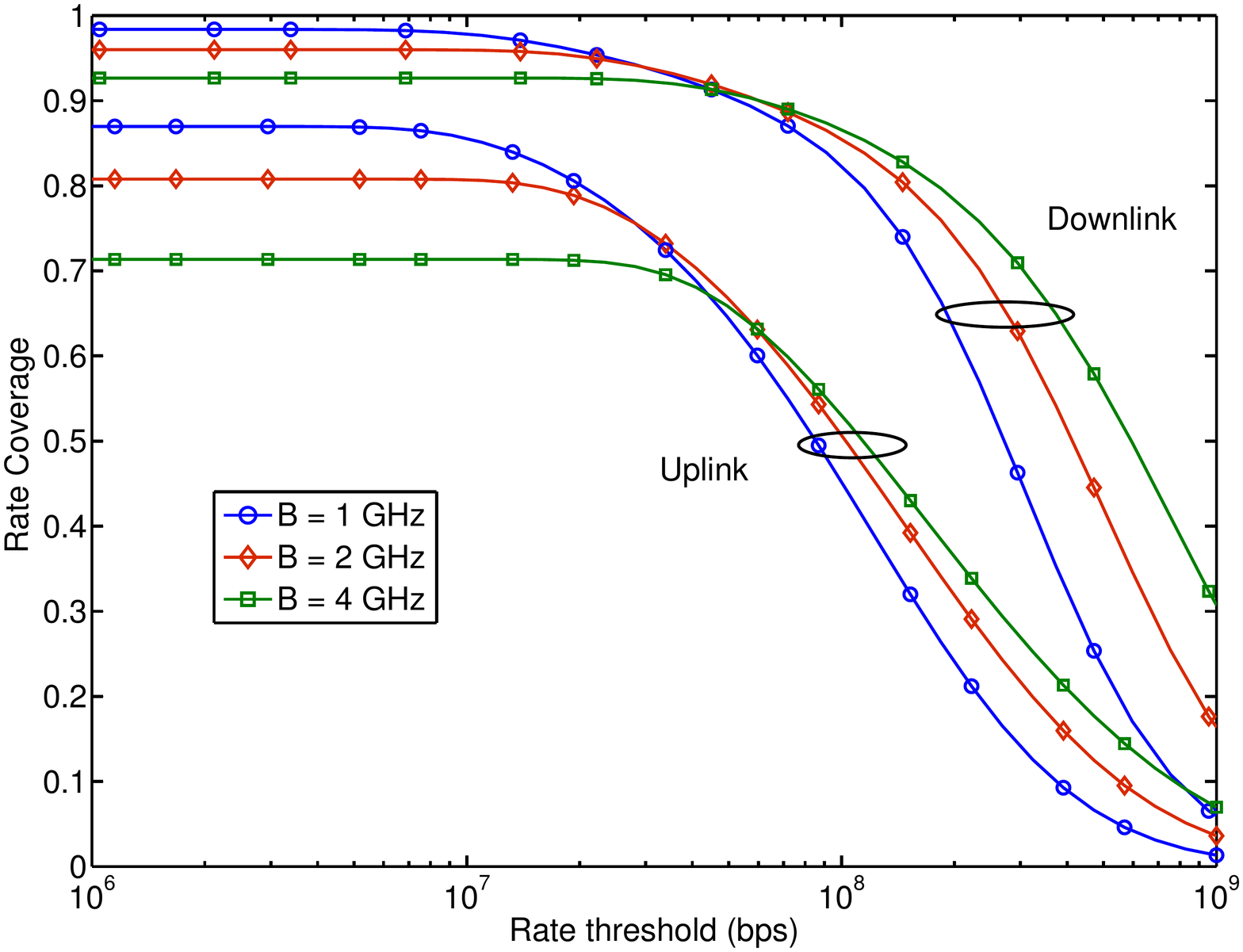}}
\caption{  Effect of bandwidth and min $\SNR$ constraint ($\rate =0$ for $\SNR<\SINRthresh_0$) on rate distribution  for BS density $ 100$   per sq. km.}
 \label{fig:RateCovBW}
\end{figure*} 

\subsection{Impact of co-existence }
The rate distribution of  a  mmWave only network and that of a mmWave-UHF hybrid network is shown in Fig. \ref{fig:HybridDens} for  different mmWave BS  densities and fixed UHF network density of  $\mu = 5$ BS per sq. km. The path loss exponent for the UHF link is assumed to equal $4$ with lognormal shadowing  of $8$ dB standard deviation.  Offloading users from mmWave  to UHF, when the link $\SNR$ drops below $\tmin=-10$ dB improves the rate of edge users significantly,  when the min  $\SNR$ constraint ($\SINRthresh_0=-10$ dB) is imposed. Such  gain from co-existence, however, reduces  with increasing mmWave BS density, as the fraction of ``poor" $\SNR$ users reduces.  Without any such minimum $\SNR$ consideration, i.e., $\SINRthresh_0 =0$, mmWave is   preferred due to the $100$x larger bandwidth.  So the key takeaway here is that users should be offloaded to a co-existing UHF macrocellular network only when reliable communication over the mmWave link is unfeasible.
\begin{figure}
 	\centering
 		\includegraphics[width=\columnwidth]{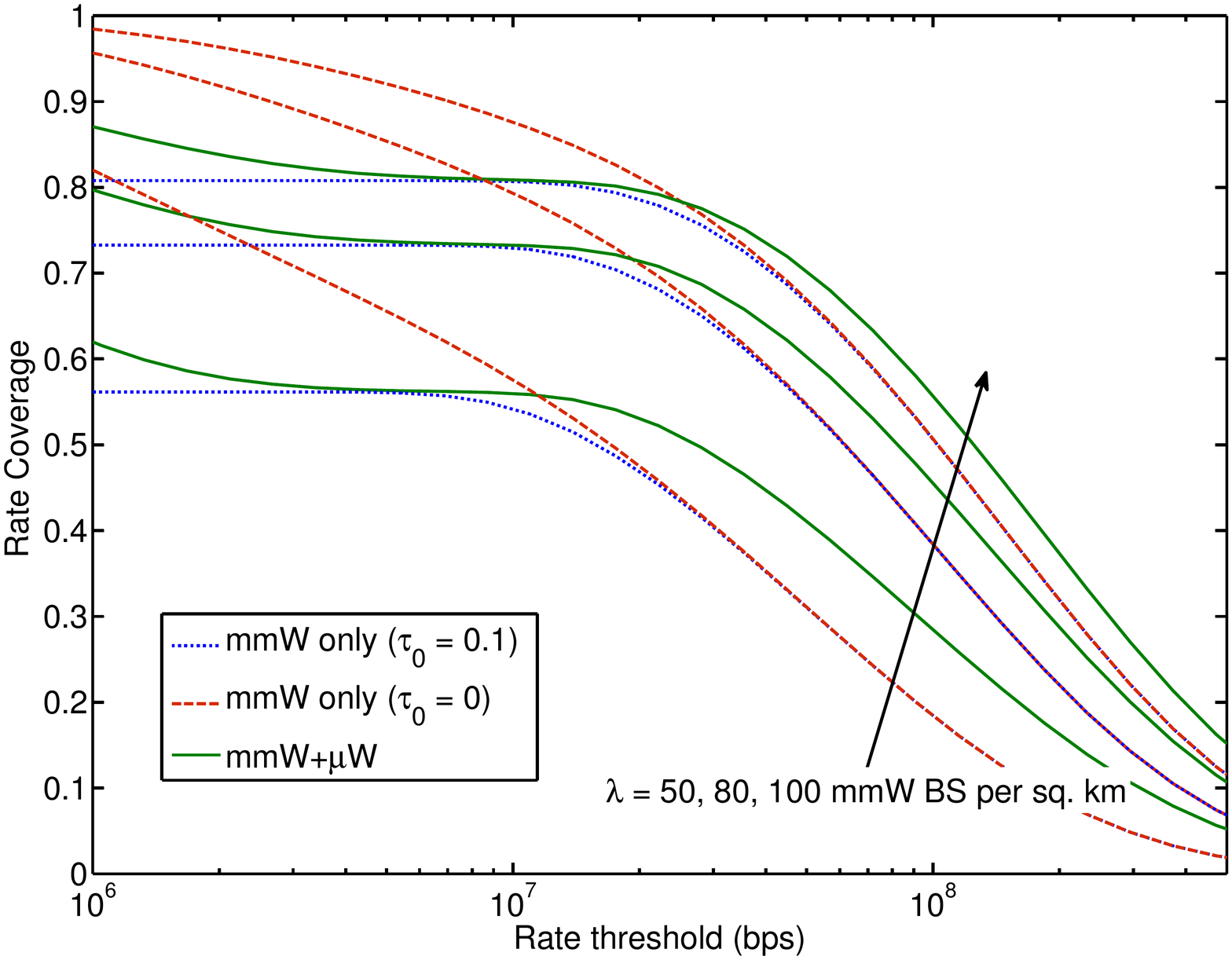}
 		\caption{Downlink rate distribution for mmWave only and hybrid network  for different mmWave BS density and  fixed UHF density of 5 BS per sq. km.}
 	\label{fig:HybridDens}
 \end{figure}

\subsection{ Impact of self-backhauling}\label{sec:sb}
The variation of downlink rate distribution with the fraction of A-BSs $\wfrac$   in the network with  BS density of $100$ and $150$ per sq. km is shown in Fig. \ref{fig:RateCovVarWfrac}. As can be seen,  providing wired backhaul to increasing fraction of  BSs improves the overall rate distribution. However ``diminishing return"   is seen with  increasing $\wfrac$ as the bottleneck shifts from the backhaul to the air interface rate. Further, it can be observed from the plot that different combinations of A-BS fraction and BS density, e.g.  ($\wfrac=0.25$, $\lambda=150$)  and ($\wfrac=0.5$, $\lambda=100$) lead to similar rate distribution.   This is investigated  further using Lemma \ref{lem:rcov} in Fig. \ref{fig:QoSContour}, which characterizes the different contours of ($\wfrac$, $\dnsty$) required to guarantee various median rates $\RATEthresh_{50}$ ($\rcov(\RATEthresh_{50})=0.5$)  in the network. For example, a median rate of $400$ Mbps in the network can be provided by either $\wfrac = 0.9, \dnsty=110$ or $\wfrac = 0.3, \dnsty=200$. Thus, the  key insight from these results is that it is feasible to provide the same QoS (median rate here) in the network by either  providing wired backhaul to a small  fraction of BSs in a dense  network, or by increasing the corresponding fraction  in  a sparser network.
\begin{figure}
 	\centering
 		\includegraphics[width=\columnwidth]{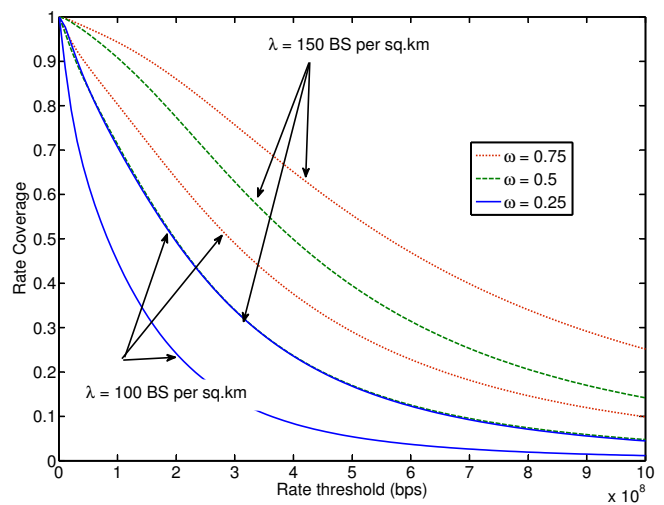}
 		\caption{Rate distribution with variation in $\wfrac$}
 	\label{fig:RateCovVarWfrac}
 \end{figure}
 \begin{figure}
 	\centering
 		\includegraphics[width=\columnwidth]{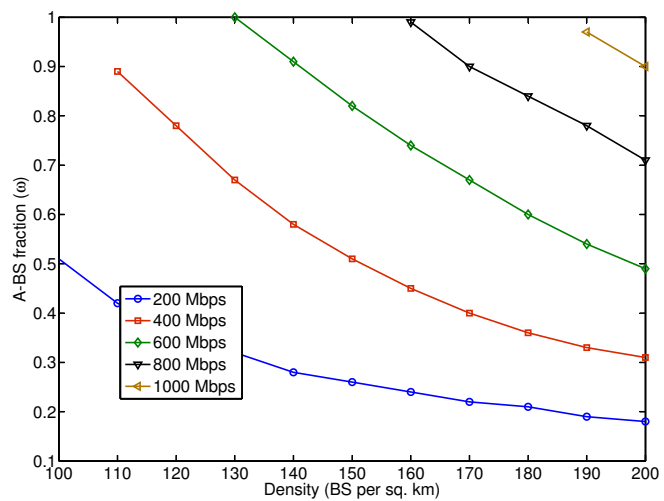}
 		\caption{The required $\wfrac$ for achieving different median rates with varying density}
 	\label{fig:QoSContour}
 \end{figure}
In the above plots,  the actual number of A-BSs in a given area increased with increasing  density for a fixed $\wfrac$, but if  the  density of A-BSs is fixed  ($\gamma$, say) while increasing the density of BSs, i.e., $\wfrac =\frac{\gamma}{\dnsty}$ for some constant $\gamma$, would a similar trend as the earlier plot be seen?  This can be answered by a closer look at Lemma \ref{lem:rcov}.  With increasing $\dnsty$, the rate coverage   of the access link increases shifting  the bottleneck to backhaul link, which in turn is limited by the A-BS density. This notion is formalized in the following proposition.
\begin{prop}
We define the saturation density $\dnsty_{\mathrm{sat}}^\delta(\gamma)$   as the density beyond which only marginal ($\delta\%$ at most) gain in rate coverage can be obtained with A-BS density fixed at $\gamma$, and characterized as 
\small
\begin{align}\label{eq:dsat}
 \arg\inf_{\dnsty}\left\{ \|\scov_d\left(\uRATEthresh\left\{\nRATEthresh 1.28\frac{\userdnsty}{\dnsty} \right\}\right) -1\|\leq \delta/\scov_b\left(\uRATEthresh\left\{\nRATEthresh 1.28^2 \frac{\userdnsty}{\gamma}\right\}\right)\right\}.
\end{align}\normalsize
\end{prop}
\begin{proof}
As the contribution from the access rate coverage can be at most $1$, the saturation density is characterized from Corollary \ref{cor:meanrcov} as 
\footnotesize
\begin{align*}
\dnsty_{\mathrm{sat}}^\delta(\gamma): & \arg\inf_{\dnsty}\Bigg\{ |\scov_d\left(\uRATEthresh\left\{\nRATEthresh \left(1+1.28\frac{\userdnsty}{\dnsty}\right) \frac{2\gamma+1.28(\dnsty-\gamma)}{\gamma+1.28(\dnsty-\gamma)}\right\}\right) -1| \nonumber\\&\leq \delta\scov_b\left(\uRATEthresh\left\{\nRATEthresh\left(1+1.28\frac{\userdnsty}{\dnsty}\right)\left(2+1.28\frac{\dnsty-\gamma}{\gamma}\right)\right\}\right)^{-1}\Bigg\}.
\end{align*}\normalsize
Noticing $\dnsty >> \gamma$ and $\userdnsty >> \dnsty$ leads to  the result.
\end{proof}
From (\ref{eq:dsat}), it is clear that $\dnsty_{\mathrm{sat}}^\delta(\gamma)$ increases with $\gamma$, as RHS decreases. For various values of A-BS density, Fig. \ref{fig:RateCovFixWDens} shows the variation in rate coverage with BS density  for a rate threshold of $100$ Mbps.  As postulated above, the rate coverage saturates with increasing density for each A-BS density.  Also shown is the saturation density obtained from (\ref{eq:dsat}) for a margin $\delta$ of $2\%$. Further, saturation density is seen to be increasing with the A-BS density, as more BSs are   required for access rate to dominate the increasing backhaul rate. 
\begin{figure}
 	\centering
 		\includegraphics[width=\columnwidth]{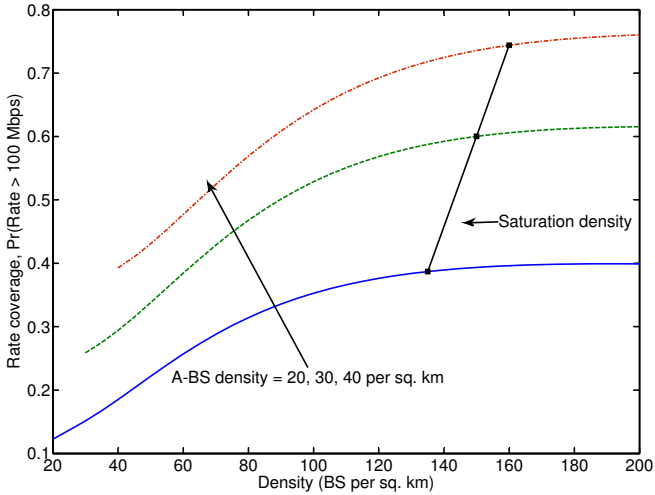}
 		\caption{Rate distribution with variation in BS density but fixed A-BS density. }
 	\label{fig:RateCovFixWDens}
 \end{figure}

\section{Conclusion and future challenges}
A baseline model and analytical framework is presented for characterizing the rate distribution in mmWave cellular networks. To the best of authors' knowledge,  the presented work is the first to integrate self-backhauling among BSs and co-existence with a  conventional  macrocellular  network into  the analysis  of mmWave networks.  We show that bandwidth plays minimal impact on the rate of power and noise-limited cell edge users, whereas increasing the BS density improves the corresponding rates drastically.  This paper also further establishes the noise-limited nature of  large bandwidth narrow beam-width mmWave networks. With self-backhauling, the rate saturates with increasing  BS density for fixed A-BS density, where the corresponding saturation density is directly proportional to the A-BS density.   The explicit characterization of the rate distribution as a function of key system parameters, which we provide, should help advance further  the understanding of such networks and benchmark their performance.

The presented work can be extended in a number of directions. Offloading of  indoor users, which  may not even   receive the signal from outdoor mmWave BSs, to   more stable networks like $4$G or WiFi  could be further investigated.     Allowing multihop backhaul in sparser deployment of A-BSs could also be investigated in future work. The developed analytical framework   also provides tools to analyze other network architectures like device-to-device (D2D) and ad hoc mmWave networks.

\section*{Acknowledgment}
The authors appreciate feedback from  Xinchen Zhang. 
\appendices
\section{}\label{sec:proofpldist}
\begin{IEEEproof}[Derivation of path loss distribution]
We drop the subscript `a' for access in this proof. The propagation process    $\PropPP\coloneqq \pl(X)= \sha(X)^{-1}\|X\|^{\rple(X)}$ on $\real{}$ for $X \in \PPP$, where $\sha \triangleq 10^{-(\plrnd+\fsldb)/10}$,  has the intensity measure  \[
\dnstyr((0,t])  = \int_{\real{2}} \pr(\pl(X) <t) \mathrm{d}X = 2\pi\dnsty \int_{\real{+}} \pr \left( \frac{r^{\rple(r)}}{\sha(r)} < t\right) r\mathrm{d}r.\]
Denote a link  to be of type $j$, where $j =l$ (LOS) and $j=n$ (NLOS) with probability $\plos_{j,\dlos}$ for link length less than $\dlos$ and $\plos_{j,\bar{\dlos}}$ otherwise. Note by construction $\plos_{l,\dlos} + \plos_{n,\dlos}=1$ and $\plos_{l,\bar{\dlos}} + \plos_{n,\bar{\dlos}}=1$. The intensity measure is then
\small
\begin{align*}
&\dnstyr((0,t])    = 2\pi\dnsty \sum_{j \in \{l,n\}} \plos_{j,\dlos}\int_{\real{+}} \pr \left( \frac{r^{\ple_j}}{\sha_j} < t\right)\indic(r < \dlos)r\mathrm{d}r \\&+  \plos_{j,\bar{\dlos}}\int_{\real{+}} \pr \left( \frac{r^{\ple_j}}{\sha_j} < t\right)\indic(r > \dlos)r\mathrm{d}r\\
& = 2\pi\dnsty\,\mathbb{E}\Bigg[ \sum_{j \in \{l,n\}} (\plos_{j,\dlos}-\plos_{j,\bar{\dlos}})\frac{\dlos^2}{2}\indic(\sha_j> \dlos^{\ple_j}/t) \\&+ \plos_{j,\dlos}\frac{(t\sha_j)^{2/\ple_j}}{2}\indic(\sha_j <\dlos^{\ple_j}/t) +  \plos_{j,\bar{\dlos}}\frac{(t\sha_j)^{2/\ple_j}}{2}\indic(\sha_j >\dlos^{\ple_j}/t)\Bigg]\\ &
 =\dnsty\pi \sum_{j \in \{l,n\}}( \plos_{j,\dlos}- \plos_{j,\bar{\dlos}} )\dlos^2 \bar{F}_{\sha_j}(\dlos^{\ple_{j}}/t)  \\&+ t^{2/{\ple_j}}  \left(\plos_{j,\dlos}{\utmoment}_{\sha_j,2/{\ple_j}}(\dlos^{\ple_j}/t) + \plos_{j,\bar{\dlos}} {\ltmoment}_{\sha_j,2/\ple_j}(\dlos^{\ple_j}/t)\right),
\end{align*}
\normalsize
where  $\bar{F}_{\sha} $ denotes the CCDF of $\sha$, and ${\utmoment}_{\sha,n}(x), {\ltmoment}_{\sha,n}(x)$ denote the truncated $\uth{n}$ moment of $\sha$ given by 
$ \utmoment_{\sha,n}(x) \triangleq \int_{0}^{x} s^n f_{S}(s) \mathrm{d}s $ and $ \ltmoment_{\sha,n}(x) \triangleq \int_{x}^{\infty} s^n f_{S}(s) \mathrm{d}s $.  
Since  $\sha$  is a Lognormal random variable $\sim \mathrm{ln}\mathcal{N}(m,\sigma^2)$, where $m  = -0.1\fsldb\ln10$ and $\sigma =0.1\sdpl \ln10$. The intensity measure in Lemma \ref{lem:pldist} is then obtained by using
\begin{align*}
\bar{F}_{\sha}(x) & = \Q\left(\frac{\ln x-m}{\sigma}\right), \\
\utmoment_{\sha,n}(x)  &  =  {\exp(\sigma^2 n^2/2+m n)}\Q\left(\frac{\sigma^2 n - \ln x+m}{\sigma}\right) \\
\ltmoment_{\sha,n}(x) & = {\exp(\sigma^2 n^2/2+m n)}\Q\left(-\frac{\sigma^2 n - \ln x+m}{\sigma}\right).
\end{align*}
Now, since $\PropPP$ is a PPP, the  distribution of path loss to the tagged BS  is then $\pr(\inf_{X\in \PPP} \pl(X)> t) = \exp(-\dnstyr((0,t]))$.
\end{IEEEproof}

\section{}\label{sec:proofulsinr}
\begin{IEEEproof}[Uplink $\SNR$ with fractional power control]
With fractional power control, a user transmits with a power $\power_u = \power_0\pl_a^\pcf$ that partially compensates for path  loss $\pl$,  where $0\leq \pcf \leq 1$  is the power control fraction (PCF) and $\power_0$ is the open loop power parameter. 
In this case, the uplink $\SNR$ CCDF is 
\begin{align*}
\pr(\SNR_u > \SINRthresh) &  = \pr \left(\frac{ \power_0 \antgainmax  \pl_a(\sbs)^{\pcf-1}}{\noisepower} > \SINRthresh\right)\\&
 = 1 - \exp\left(-\dnsty \ndnstyr_a\left(\left(\frac{ \power_0\antgainmax}{\SINRthresh\noisepower}\right)^{1/(1-\pcf)}\right)\right).
\end{align*}

\end{IEEEproof}

\section{}\label{sec:sinrproof}
\begin{IEEEproof}[$\SINR$ distribution]
Having derived the intensity measure of $\PropPP$ in Lemma \ref{lem:pldist}, the distribution of $\SINR$ can be characterized on the same lines as 
\cite{BlaKarKee12}. The key steps are highlighted below for completeness.
\begin{align*}
\pr(\SINR > \SINRthresh )  = &\pr\left(\frac{\power_b  \antgainmax  \pl(\sbs )^{-1}}{\sum_{X \in \BSP \setminus\{\sbs\}}\power_b\again_X\pl(X)^{-1} +\noisepower} > \SINRthresh\right) 
\\ = & \pr\left(\IF +\frac{\noisepower\pl(\sbs )}{\power_b   \antgainmax }< \frac{1}{\SINRthresh}\right)\\
 = \int_{l>0}&\pr\left(\IF + \frac{\noisepower l}{\power_b   \antgainmax } < \frac{1}{\SINRthresh}|\pl(\sbs )=l\right)f_{\pl(\sbs)}(l)\mathrm{d} l
\end{align*}
where $\IF = \frac{\pl(\sbs)}{\antgainmax}\sum_{X \in \BSP \setminus\{\sbs\}}\again_X\pl(X)^{-1}$ and the distribution  of $\pl(\sbs)$  is derived as 
\begin{equation}\label{eq:pdfpl}
f_{\pl(\sbs)}(l) = -\frac{\mathrm{d}}{\mathrm{d}l}\pr(\pl(\sbs)>l) =\dnsty \exp(-\dnsty\ndnstyr(l))\ndnstyr^{'}(l).
\end{equation}
 The conditional CDF required  for the above computation is derived from the the conditional Laplace transform given below using the Euler's characterization \cite{AbaWar95}
\begin{align*}
\mgf_{\IF,l}(z)&=\expect{\exp(-z\IF)|\pl(\sbs)=l)}\\& = \exp\left(-\cexpect{\int_{u>l}(1-\exp(-z l \again/u))\dnstyr(\mathrm{d}u)}{\again}\right),
\end{align*}
where $\dnstyr(\mathrm{d}u)$ is given by (\ref{eq:mdt}).  

The inverse Laplace transform  calculation  required in the above  derivation could get  computationally intensive in certain cases and may render the analysis intractable. However,  introducing Rayleigh small scale fading $\chanl\sim \exp(1)$, on each link improves the tractability of the analysis as shown below. 
Coverage with fading is 
\footnotesize{
\begin{align}
&\pr\left(\frac{\power_b  \antgainmax\chanl_{\sbs} \pl(\sbs)^{-1}}{\sum_{X \in \BSP \setminus\{\sbs\}}\power_b\again_X\chanl_X\pl(X)^{-1} +\noisepower} > \SINRthresh\right) \nonumber\\
= &\expect{\exp\left(-\frac{\SINRthresh\noisepower}{\power_b  \antgainmax}\pl(\sbs)-\SINRthresh\pl_{\sbs }\sum_{X \in \BSP \setminus\{\sbs\}} \frac{\again_X}{\antgainmax} \chanl_X\pl(X)^{-1}\right)}\nonumber\\
 \overset{(a)}{=} & \int_{l>0} \exp\left(-\frac{\SINRthresh\noisepower}{\power_b  \antgainmax}l -\dnsty\cexpect{\int_{u>l}\frac{\ndnstyr^{'}(u)\mathrm{d}u}{u(zl)^{-1}+1}}{z}\right)f_{\pl_{\sbs}}(l)\mathrm{d}l\nonumber\\
 \overset{(b)}{=} & \dnsty\int_{l>0} \exp\left(-\frac{\SINRthresh\noisepower}{\power_b  \antgainmax}l-\dnsty\ndnstyr(l)\cexpect{\frac{1}{1 +z}}{\again} \right)\\&\times \left( -\dnsty\cexpect{\int_{0}^{\frac{z}{z+1}}\ndnstyr\left\{zl\left(\frac{1}{u}-1\right)\right\}\mathrm{d}u}{\again}\right)\ndnstyr^{'}(l)\mathrm{d}l\label{eq:sinrfading},
\end{align}
}\normalsize
where $z=\frac{\SINRthresh \again}{\antgainmax}$, (a) follows using the Laplace functional of point process $\PropPP$, (b) follows using integration by parts along with (\ref{eq:pdfpl}). 

The above derivation assumed all BSs to be transmitting, but since user population is finite, certain BSs may not have a user to serve with probability $1-\loadpmf(\userdnsty,\dnsty,0)$. This is incorporated in the analysis by modifying $\dnsty \to \dnsty (1- \loadpmf(\userdnsty,\dnsty,0))$ in (a) above. 
\end{IEEEproof}

\vspace*{-2\baselineskip}
\begin{IEEEbiography}[{\includegraphics[width=1in,height=1.25in,clip,keepaspectratio] {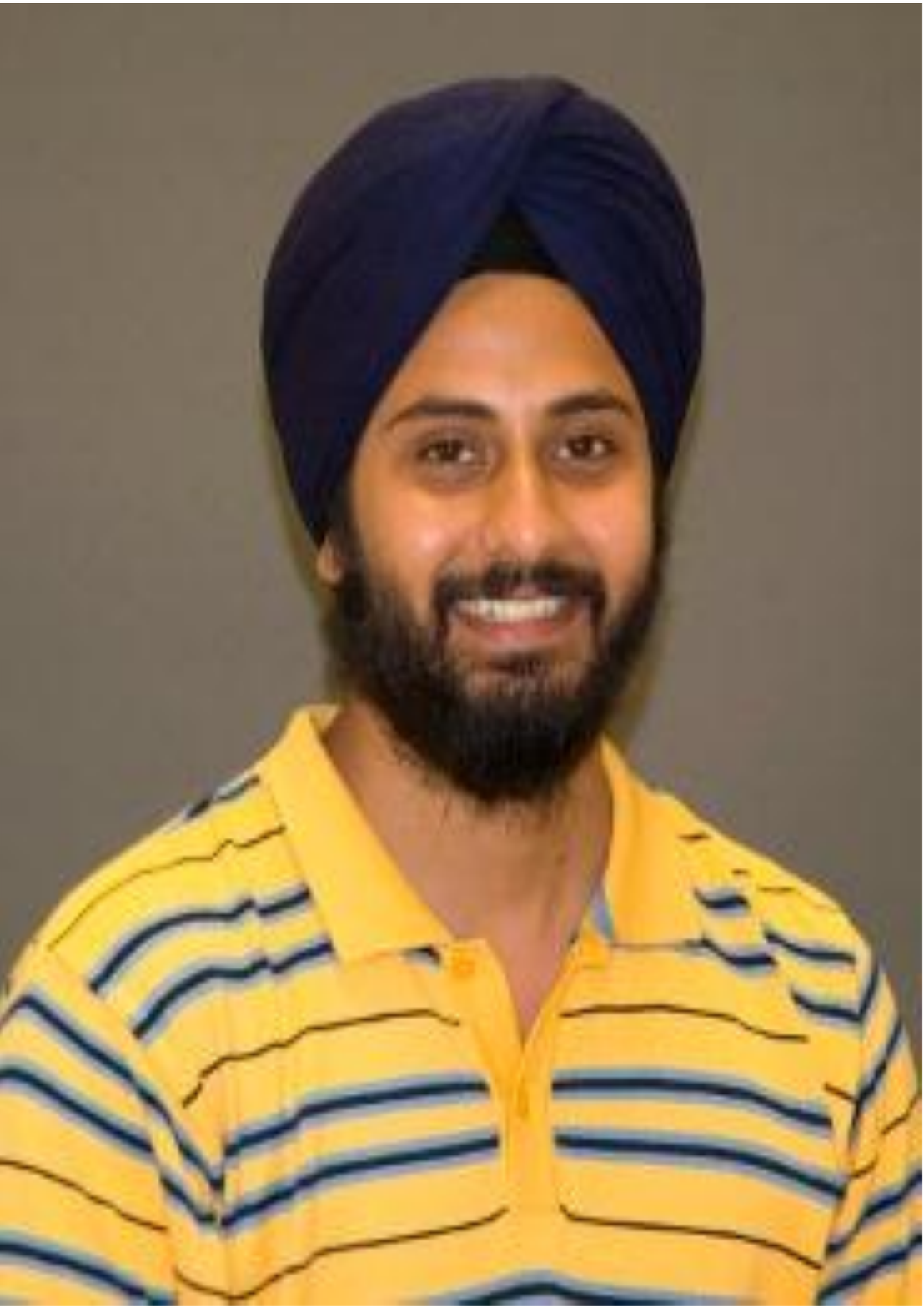}}]{Sarabjot Singh}
(S'09, M'15)  received the B. Tech.  in Electronics and Communication Engineering from Indian Institute of Technology Guwahati, India, in 2010 and  the M.S.E and Ph.D. in Electrical Engineering from University of Texas at Austin.  He is currently a Senior Researcher at Nokia Technologies, Berkeley, USA, where he is involved in the research and standardization of the next generation Wi-Fi systems. He has held industrial internships at Alcatel-Lucent Bell Labs in Crawford Hill, NJ; Intel Corp. in Santa Clara, CA; and Qualcomm Inc. in San Diego, CA. Dr. Singh was the recipient of the  President of India Gold Medal  in 2010 and the ICC Best Paper Award in 2013.  
\end{IEEEbiography}

\begin{IEEEbiography}[{\includegraphics[width=1in,height=1.25in,clip,keepaspectratio] {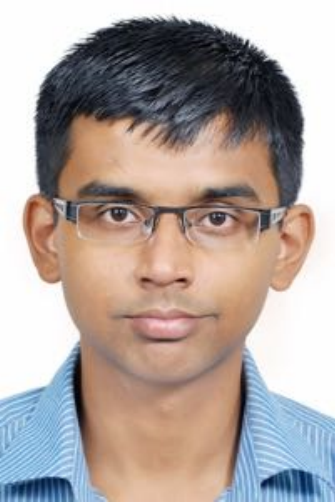}}]{Mandar Kulkarni}
(S'13) received his B.Tech degree from the Indian Institute of Technology Guwahati in 2013 and was awarded the president of India gold medal for best academic performance. He is currently working towards his PhD in Electrical Engineering at the University of Texas at Austin. His research interests are broadly in the field of wireless communication, with current focus on modeling and analysis of cellular networks operating at millimeter wave frequencies. He has held internship positions at Nokia Networks, Arlington Heights, IL, U.S.A.; Technical University, Berlin, Germany and Indian Institute of Science, Bangalore, India in the years 2014, 2012 and 2011, respectively.
\end{IEEEbiography}

\begin{IEEEbiography}[{\includegraphics[width=1in,height=1.25in,clip,keepaspectratio] {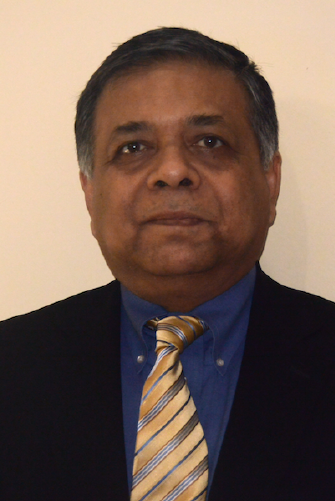}}]{Amitabha  Ghosh}
 joined Motorola in 1990 after receiving his Ph.D in Electrical Engineering from Southern Methodist University, Dallas.  Since joining Motorola he worked on multiple wireless technologies starting from IS-95, cdma-2000, 1xEV-DV/1XTREME, 1xEV-DO, UMTS, HSPA, 802.16e/WiMAX/802.16m, Enhanced EDGE and 3GPP LTE. Dr. Ghosh has 60 issued patents and numerous external and internal technical papers.  Currently, he is Head, North America Radio Systems Research within the Technology and Innovation office of Nokia Networks. He is currently working on 3GPP LTE-Advanced and 5G technologies. His research interests are in the area of digital communications, signal processing and wireless communications. He is a Fellow of IEEE and co-author of the book titled “Essentials of LTE and LTE-A”. 
\end{IEEEbiography}

\begin{IEEEbiography}[{\includegraphics[width=1in,height=1.25in,clip,keepaspectratio]{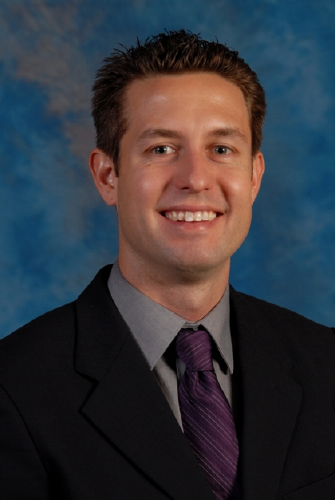}}]{Jeffrey Andrews}
(S'98, M'02, SM'06, F'13)   received the B.S. in Engineering with High Distinction from Harvey Mudd College, and the M.S. and Ph.D. in Electrical Engineering from Stanford University.  He is the Cullen Trust Endowed Professor ($\#$1) of ECE at the University of Texas at Austin, Editor-in-Chief of the IEEE Transactions on Wireless Communications, and Technical Program Co-Chair of IEEE Globecom 2014. He developed Code Division Multiple Access systems at Qualcomm from 1995-97, and has consulted for entities including Verizon, the WiMAX Forum, Intel, Microsoft, Apple, Samsung, Clearwire, Sprint, and NASA.  He is a member of the Technical Advisory Boards of Accelera and Fastback Networks, and co-author of the books Fundamentals of WiMAX (Prentice-Hall, 2007) and Fundamentals of LTE (Prentice-Hall, 2010).  

Dr. Andrews received the National Science Foundation CAREER award in 2007 and has been co-author of ten best paper award recipients: ICC 2013, Globecom 2006 and 2009, Asilomar 2008, European Wireless 2014, the 2010 IEEE Communications Society Best Tutorial Paper Award, the 2011 IEEE Heinrich Hertz Prize, the 2014 EURASIP Best Paper Award, the 2014 IEEE Stephen O. Rice Prize, and the 2014 IEEE Leonard G. Abraham Prize.  He is an IEEE Fellow and an elected member of the Board of Governors of the IEEE Information Theory Society. 

\end{IEEEbiography}
\vfill

\end{document}